\documentclass[pre,twocolumn,showpacs,superscriptaddress,preprintnumbers,floatfix]{revtex4-1}

\usepackage{etex}
\usepackage{ifpdf}
\usepackage{hyperref}
\usepackage{dcolumn}
\usepackage{amsmath}
\usepackage{amscd}
\usepackage{amsfonts}
\usepackage{amssymb}
\usepackage{bm}   
\usepackage{bbm}
\usepackage{verbatim}
\usepackage{stmaryrd}
\usepackage{amsthm}
\usepackage{xcolor}
\usepackage{setspace}
\usepackage{multirow}

\ifpdf
  \usepackage[pdftex]{graphicx}
\else
  \usepackage[dvips]{graphicx}
\fi

\theoremstyle{plain}	
\theoremstyle{plain}	
\theoremstyle{plain} 	
\theoremstyle{plain} 	
\theoremstyle{plain} 	\newtheorem{The}{Theorem}
\theoremstyle{plain} 	
\theoremstyle{plain} 	
\theoremstyle{plain} 	
\theoremstyle{plain} 	
\theoremstyle{plain}	
\theoremstyle{plain}	\newtheorem{Def}{Definition}
\theoremstyle{plain}

\newcommand{\eM}     {\mbox{$\epsilon$-machine}}
\newcommand{\eMs}    {\mbox{$\epsilon$-machines}}
\newcommand{\EM}     {\mbox{$\epsilon$-Machine}}
\newcommand{\EMs}    {\mbox{$\epsilon$-Machines}}

\newcommand{\order}[1]{order\protect\nobreakdash-$#1$}

\newcommand{\cryptic}[1]{$#1$\protect\nobreakdash-cryptic}

\newcommand{\Process}{\mathcal{P}}

\newcommand{\MeasAlphabet}	{\mathcal{A}}
\newcommand{\MeasSymbol}   { {X} }
\newcommand{\meassymbol}   { {x} }

\newcommand{\Past}	{ \smash{\overleftarrow {\MeasSymbol}} }
\newcommand{\past}	{ \smash{\overleftarrow {\meassymbol}} }

\newcommand{\Future}	{ \smash{\overrightarrow{\MeasSymbol}} }

\newcommand{\PastL}	{ \smash{\overleftarrow {\MeasSymbol}{}^L} }

\newcommand{\FutureL}	  { \smash{\overrightarrow{\MeasSymbol}{}^L} }

\newcommand{\AllPasts}	    { \smash{\overleftarrow{ {\rm {\bf \MeasSymbol}} } } }

\newcommand{\CausalState}	{ \mathcal{S} }

\newcommand{\CausalStateSet}	{ \boldsymbol{\CausalState} }
\newcommand{\AlternateState}	{ \mathcal{R} }
\newcommand{\AlternateStatePrime}	{ {\cal R}^{\prime} }
\newcommand{\alternatestate}	{ \rho }

\newcommand{\AlternateStateSet}	{ \boldsymbol{\AlternateState} }
\newcommand{\PrescientState}	{ \widehat{\AlternateState} }

\newcommand{\PrescientStateSet}	{ \boldsymbol{\PrescientState}}

\newcommand{\Prob}      {\Pr} 

\newcommand{\Cmu}		{C_\mu}
\newcommand{\hmu}		{h_\mu}
\newcommand{\EE}		{{\bf E}}
\newcommand{\TI}		{{\bf T}}
\newcommand{\SI}        {{\bf S}}

\newcommand{\PC}		{\chi}

\newcommand{\forward}{+}
\newcommand{\reverse}{-}
\newcommand{\forwardreverse}{\pm} 

\newcommand{\FutureCausalState}	{ {\CausalState}^{\forward} }

\newcommand{\PastCausalState}	{ {\CausalState}^{\reverse} }

\newcommand{\lastindex}[2]{%
  \edef\tempa{0}
  \edef\tempb{#2}
  \ifx\tempa\tempb
    \edef\tempc{#1}
  \else
    \edef\tempa{0}
    \edef\tempb{#1}
    \ifx\tempa\tempb
      \edef\tempc{#2}
    \else
      \edef\tempc{#1+#2}
    \fi
  \fi
  \tempc
}
\newcommand{\BE}[2][0]{\ensuremath{H[\MeasSymbol_{#1}^{#2}]}}
\newcommand{\BSE}[2][0]{
  \ensuremath{H[\MeasSymbol_{#1}^{#2},\CausalState_{\lastindex{#1}{#2}}]}%
}
\newcommand{\SBE}[2][0]{\ensuremath{H[\CausalState_{#1},\MeasSymbol_{#1}^{#2}]}}

\newcommand{\AltSBE}[2][0]{\ensuremath{H[\AlternateState_{#1},\MeasSymbol_{#1}^{#2}]}}
\newcommand{\AltBSE}[2][0]{\ensuremath{H[\MeasSymbol_{#1}^{#2},\AlternateState_{\lastindex{#1}{#2}}]}}

\newcommand{\GI}{\varphi}
\newcommand{\OI}{\zeta}
\newcommand{\COrder}{k_{\PC}}
\newcommand{\GOrder}{k_{\GI}}
\newcommand{\OOrder}{k_{\OI}}
\newcommand{\MOrder}{R}
\newcommand{\SOrder}{k_{\SI}}

\newcommand{\CSjoint}[1][,]{
   \edef\tempa{:}
   \edef\tempb{#1}
   \ifx\tempa\tempb
      \ensuremath{\FutureCausalState\!#1\PastCausalState}
   \else
      \ensuremath{\FutureCausalState#1\PastCausalState}
   \fi
}

\newif\ifpm
\edef\tempa{\forwardreverse}
\edef\tempb{\pm}
\ifx\tempa\tempb
   \pmfalse
\else
   \pmtrue
\fi

\newif\ifcolor

\colortrue

\ifcolor\relax\else\selectcolormodel{gray}\fi

\begin{document}

\title{Synchronization and Control in
Intrinsic and Designed Computation:\\
An Information-Theoretic Analysis of\\
Competing Models of Stochastic Computation}

\author{James P. Crutchfield}
\email{chaos@ucdavis.edu}
\affiliation{Complexity Sciences Center and Physics Department,
University of California at Davis, One Shields Avenue, Davis, CA 95616}
\affiliation{Santa Fe Institute, 1399 Hyde Park Road, Santa Fe, NM 87501}

\author{Christopher J. Ellison}
\email{cellison@cse.ucdavis.edu}
\affiliation{Complexity Sciences Center and Physics Department,
University of California at Davis, One Shields Avenue, Davis, CA 95616}

\author{Ryan G. James}
\email{rgjames@ucdavis.edu}
\affiliation{Complexity Sciences Center and Physics Department,
University of California at Davis, One Shields Avenue, Davis, CA 95616}

\author{John R. Mahoney}
\email{jrmahoney@ucdavis.edu}
\affiliation{Complexity Sciences Center and Physics Department,
University of California at Davis, One Shields Avenue, Davis, CA 95616}

\date{\today}

\bibliographystyle{unsrt}

\begin{abstract}
We adapt tools from information theory to analyze how an observer comes to
synchronize with the hidden states of a finitary, stationary stochastic process.
We show that synchronization is determined by both the process's internal
organization and by an observer's model of it. We analyze these components
using the convergence of state-block and block-state entropies, comparing
them to the previously known convergence properties of the Shannon block
entropy. Along the way, we introduce a hierarchy of information quantifiers as
derivatives and integrals of these entropies, which parallels a similar
hierarchy introduced for block entropy. We also draw out the duality between
synchronization properties and a process's controllability. The tools lead
to a new classification of a process's alternative representations in terms
of minimality, synchronizability, and unifilarity.

\vspace{0.1in}

\end{abstract}

\preprint{Santa Fe Institute Working Paper 10-07-XXX}
\preprint{arxiv.org:1007.XXXX [physics.gen-ph]}

\maketitle
\

\tableofcontents

\setstretch{1.1}

\vspace{0.3in}

{\bf Nonlinear dynamical systems store and generate information---they
\emph{intrinsically compute}.
Real computing devices use nonlinearity to do the same, except that they are
\emph{designed to compute}---the information serves some utility or function
determined by the designer. Intuitively, useful computing devices must be
constructed out of (physical, chemical, or biological) processes that have some
minimum amount of intrinsic computational capability. However, the exact
relationship between intrinsic and designed computation remains elusive. In
fact, bridging intrinsic and designed computation requires solving a number of
intermediate problems. One is to understand the diversity of intrinsic
computation of which nonlinear dynamical systems are capable. Another is to
determine if one can practically manipulate these systems in the service of
functional information generation and storage.

Here, we address both of these problems from the perspective of information
theory. We describe new information processing characteristics of dynamical
systems and the stochastic processes they generate. We focus particularly on
two key aspects that impact design: synchronization and control.
Synchronization concerns how we come to know the hidden states of a process
through observations; while control, how we can manipulate a process into a
desired internal condition.
} 

\section{Introduction}

Given a model of a stationary stochastic process, how much information must
one extract from observations to exactly know which state the process is in?
With this, an observer is said to be synchronized to the process.
(For an introduction to the problem, see Ref.~\cite{Crut01b}.)

Given that one has designed a stochastic process, is there a series of inputs
that reliably drive it to a desired internal condition? If so, the designed
process is said to be controllable.

Synchronization and control are dual to each other: In synchronization, an
observer attempts to predict the process's internal state from incomplete and
indirect observations, typically starting with complete ignorance and
hopefully ending with complete certainty. In control, one must extract from
the design a series of manipulations, typically indirect, that will drive the
process to a desired state or set of states. The duality is simply that the
observer's measurements can be interpreted as the designer's control inputs.

Synchronization and control are key aspects in intrinsic and designed
computation, both to detecting intrinsic computation in dynamical systems
and to leveraging a dynamical system's intrinsic computation into useful
computation. For the latter, the circuit designer attempts to build circuits,
themselves dynamical systems, that synchronize to incoming signals.


For example,
even the most mundane initial operation is essential: When power is first
applied, a digital computer must predictably reach a stable and repeatable
state, without necessarily being able to perform even small amounts of digital
intelligent control or analysis of its changing environment. Without reliably
reaching a stable condition---now a quite elaborate operation in modern
microprocessors---no useful information processing can be initiated. The
device is still a dynamical system, of course, but it fails at raising itself
from that prosaic condition to the level of a computing device.

Once digital computing operations have commenced, similar concerns arise in the
timing and control of information being loaded from memory into a register.
Not only must each data bus line synchronize properly or risk misconstruing
the voltage level offered up by the wires, but this must happen simultaneously
across a number of component devices---quite wide buses, 128 and 256 lines
are not uncommon today.

Stepping back a bit, one must wonder what tools dynamical systems theory itself
provides to analyze and design computation. Indeed, many of the properties
often used to characterize and classify dynamical systems are
time-asymptotic---the Kolmogorov-Sinai entropy or Shannon entropy rate,
the spectrum of Lyapunov characteristic exponents, the fractal and information
dimensions (which rely on the asymptotic invariant measure), come to mind.
However, real computing is not asymptotic. Individual logic gates, as dynamical
systems, deliver their results on the short term. Indeed, the faster they do
this, the better.

How can we bridge the gap between dynamical systems theory and
the need to characterize the short term properties of dynamical systems?
A suggestive example is found in the analysis of escape rates~\cite{Wigg92a},
a property of transient, short-term behavior. Another answer is found in
synchronization and controllability, as they too are properties of the
short term behavior of dynamical systems. We will show that there is a
connection between these properties and the more typical asymptotic view of dynamical
behavior: Synchronization and control are determined by the nature of
convergence to the asymptotic---they are our subject.

Given the duality between synchronization and control, in the following we
present results in terms of only one notion---synchronization. The results
apply equally well to control, though with different interpretations.

\subsection{Precis}

Analyzing informational convergence properties is the main strategy we will use.
However, as we will see, different properties converge differently
from each other, either for a given process or as one looks across a family of
processes. Moreover, for a given process we will consider a family of
different representations of it. The result, while giving insight into
informational properties and how representations can distort them,
ends up being a rather elaborate classification scheme. To reduce the
apparent complication, it will be helpful to give a detailed summary of
the steps we employ in the development.

After describing related work, we review the use of Shannon block entropy and
related quantities, analyzing their asymptotic behavior and aspects of convergence.
We introduce a single framework---the convergence hierarchy---to call out the systematic nature of
convergence properties.

We then take a short detour to introduce the range of possible descriptions a
process can have, noting their defining properties. One, the \eM, plays a
particularly central role, as it allows one to calculate all of a process's
intrinsic properties. Other descriptions typically do not allow this to the
same broad extent.

With a model in hand, one can start to discuss how one synchronizes to its
states. When the model is the \eM, one can speak of synchronizing to
the process itself. To do this, we analyze the convergence properties of two
new entropies: the state-block entropy and the block-state entropy. We
establish their general asymptotic properties, introducing convergence
hierarchies of their own, paralleling that for the block entropy. For finitary
processes, the latter converges from below, but the new block-state entropy
converges from above to the same asymptote. One benefit is that estimation
methods can be improved through use of bounds from above and below.

When we specialize to the \eM, we establish a direct connection between
synchronization and how the block entropies converge. We provide an
informational measure---synchronization information---that summarizes the total
uncertainty encountered
during synchronization. We relate this back to the transient information
introduced previously, which derives only from the observed sequences,
not requiring a model or a notion of state. Along the way, we discuss a
process's Markov
order---the scale at which ``asymptotic'' statistics set in---and its cryptic
order---the length scale over which internal state information is spread.
These scales control synchronization.

The development then, step-by-step, relaxes the \eM's defining properties
in order to explore an increasingly wide range of models. A particular
emphasis in this is to show how nonoptimal models bias estimates of a
process's informational properties. Conversely, we learn how certain classes
of models, some widely used in mathematical statistics and elsewhere, make
strong assumptions and, in some cases, preclude the estimation of important
process properties.

Starting with the class of minimal, optimally predictive models that synchronize
(finitary \eMs), we first relax the minimality assumption. We show that needless
model elaborations---such as more, but redundant states---can affect
synchronization. We identify that class which still does
synchronize. Then, we consider nonminimal unifilar,
nonsynchronizing models. Finally, we relax the unifilarity assumption.
At each stage, we see how the convergence properties of the various entropies
change. These changes, in turn, induce a number of informational measures
of what the models themselves contribute to a process's now largely-apparent
information processing.

A key tool in the analysis takes advantage of the fact that the various
multivariable information quantities form a signed measure~\cite{Yeun91a}.
Their visual display, a form of Venn diagram called an information diagram,
brings some order to the notation and classification chaos.

\subsection{Synchronization and Control: Related Work}

Controlling dynamical systems and stochastic processes has an extensive history.
For linear dynamical systems see, for example, Ref.~\cite{Klam91a} and for
hidden Markov models see, for example, Ref.~\cite{Elli94a}. More recently,
there has been much work on controlling nonlinear dynamical systems, a markedly
more difficult problem in its full generality; see
Refs.~\cite{Andr03a,Andr04a,Gonz04a}.

Synchronization, too, has been very broadly studied and for much longer, going
back at least to Huygens~\cite{Piko01a}. It is also an important property of
symbolic dynamical systems~\cite{Jono96a}. It has even become quite popularized
of late, being elevated to a general principle of natural
organization~\cite{Stro03a}.

Here, we consider a form of synchronization that is, at least at this point,
distinct from the dynamical kind.
Moreover, we take a complementary, but distinct approach---that of information
theory---to address control and synchronization. This was introduced in
Ref.~\cite{Crut01a} and several applications are given in
Refs.~\cite{Crut01b,Feld02a}. A roughly similar problem setting for
synchronization is found in Ref.~\cite{Mark73a}. We note that the closely
related topics of state estimation and control are addressed in information
theory~\cite{Feng97a,Ahme98a}, nonlinear dynamics~\cite{Pack80,Take81,Ott04a},
and Markov decision processes~\cite{Pute05a}.

Adapting the present approach to continuous dynamical systems and stochastic
processes remains a future effort. For the present, the closest connections will
be found to the work cited above on hidden Markov models and symbolic dynamical
systems.

\vspace{-0.1in} 
\section{Block Entropy and Its Convergence Hierarchy}

It is an interesting fact, perhaps now intuitive, that to estimate even the
randomness of an information source, one must also estimate it's internal
structure. Ref.~\cite{Crut01a} gives a review of this interdependence and
it serves as a starting point for our analysis of synchronization, which is
a question about coming to know the source's states from observations.
Indeed, if one has to make estimates of internal organization just to get to
randomness, then one, in effect and without too much more effort, can also
address issues of synchronization. There is an intimate relationship that
we hope to establish.

We briefly review Ref.~\cite{Crut01a}, largely to introduce notation and
highlight the main ideas needed for synchronization. This review and our
development of synchronization requires the reader to be facile with information
theory at the level of the first half of Ref.~\cite{Cove06a}, signed
information measures and information diagrams of Ref.~\cite{Yeun91a},
and their uses in Refs.~\cite{Crut08a,Crut08b,Maho09a}.

\vspace{-0.1in} 
\subsection{Stationary Stochastic Processes}

The approach in Ref.~\cite{Crut01a} starts simply: Any stationary process, $\Process$, is
a joint probability distribution $\Pr(\Past,\Future)$ over past and future
observations. This distribution can be thought of as a
\emph{communication channel} with a specified input distribution, $\Pr(\Past)$.
It transmits information from the \emph{past}
$\Past = \ldots \MeasSymbol_{-3} \MeasSymbol_{-2} \MeasSymbol_{-1}$ to the
\emph{future} $\Future = \MeasSymbol_0 \MeasSymbol_1 \MeasSymbol_2 \ldots$
by storing it in the present. $\MeasSymbol_t$ is the random variable for
the measurement outcome at time $t$; the lowercase $\meassymbol_t$ denotes a
particular value. Throughout this work, we always use $\Past$ and $\Future$ in
the limiting sense. That is, we work with length-$L$ sequences or \emph{blocks}
of random variables:
$\MeasSymbol_t^L = \MeasSymbol_t\MeasSymbol_{t+1} \cdots \MeasSymbol_{t+L-1}$
and take the limit as $L$ approaches infinity.

In the following, we consider only discrete measurement
outcomes---$\meassymbol \in \MeasAlphabet = \{1,2,\ldots,k\}$---and stationary
processes---$\Prob(\MeasSymbol_t^L) = \Prob(\MeasSymbol_0^L)$, for all times
$t$ and block lengths $L$. Unlike some definitions of stationarity, this makes
no assumptions about the process's internal starting conditions, as such
knowledge obviates the very question of synchronization.

Such processes include those found in the field of stochastic processes, of
course, but one also has in mind the symbolic dynamics of continuous-state
continuous-time or continuous-state discrete-time dynamical systems on their
invariant sets. The notions also apply equally well to one-dimensional spatial
configurations of spin systems and of deterministic and probabilistic cellular
automata, where one interprets the spatial coordinate as a ``time''.

\vspace{-0.1in}
\subsection{Block Entropy}

One measure of the diversity of length-$L$ sequences generated by a process is
its \emph{Shannon block entropy}:
\begin{align}
H(L) & \equiv H[\MeasSymbol_0^L] \\
     & = - \sum_{w \in \MeasAlphabet^L} \Prob(w) \log_2 \Prob(w) ~,
\label{eq:BlockEntropy}
\end{align}
where $w = \meassymbol_0 \meassymbol_1 \ldots \meassymbol_{L-1}$ is a
\emph{word} in the set $\MeasAlphabet^L$ of length-$L$ sequences.
It has units of [bits] of information.
One can think of the block entropy as a kind of transform that reduces a
process's distribution over the (typically infinite) number of sequences to
a function of a single variable $L$. In this view, Ref.~\cite{Crut01a} focused
on a simple question: What properties of a process can be determined solely
from its $H(L)$?

\vspace{-0.1in} 
\subsection{Source Entropy Rate}

One of those properties, and historically the most widely used and
technologically important, is \emph{Shannon's source entropy rate}:
\begin{align}
\hmu = \lim_{L \rightarrow \infty} \frac{H(L)}{L} ~.
\label{eq:EntropyRate}
\end{align}
The entropy rate is the irreducible unpredictability of a process's
output---the intrinsic randomness left after one has extracted all of the
correlational information from past observations. The difference between it and
the alphabet size, $\log_2 |\MeasAlphabet| - \hmu$, indicates how much the raw
measurements can be compressed. More precisely, Shannon's First Theorem states
that the output sequences $\meassymbol^L$ from an information source can be
compressed, without error, to $L \hmu$ bits \cite{Cove06a}. Moreover, Shannon's
Second Theorem gives operational meaning to the entropy rate \cite{Cove06a}:
A communication channel's capacity must be larger than $\hmu$ for error-free
transmission.

\vspace{-0.1in} 
\subsection{Excess Entropy}

As noted, any process---chaotic dynamical system, spin chain, cellular automata,
to mention a few---can be considered a channel that communicates its past to
its future. The messages to be transmitted in this way are the pasts which the
process can generate. Thus, the ``capacity'' of this channel is not something
that one optimizes as done in Shannon's theory to engineer channels and
construct error-free encodings. Rather, we think of it as how much of the
process's channel is actually used.

A process's channel utilization is another property that can be determined from
the block entropy. It is called the \emph{excess entropy} and is defined,
closely following Shannon's channel capacity definition, by:
\begin{align}
\EE = I[\Past;\Future] ~,
\label{eq:ExcessEntropyMI}
\end{align}
where $I[Y;Z]$ is the mutual information between random variables $Y$ and $Z$.
It has units of [bits] and tells one how much information the output shares
with the input and so measures how much information is transmitted through a,
possibly noisy, channel.

\vspace{-0.1in} 
\subsection{Block Entropy Asymptotics}

It has been known for quite some time now that the entropy rate and excess
entropy control the asymptotic behavior ($L \rightarrow \infty$) of a
\emph{finitary} process's block entropy. Specifically, it scales according to
the linear asymptote:
\begin{align}
H(L) \propto \EE + \hmu L ~.
\label{eq:beasymptote}
\end{align}
Specifically,
\begin{align}
\EE = \lim_{L \rightarrow \infty} \left( H(L) - L\hmu \right) ~.
\label{eq:EEAsymptoteDefn}
\end{align}
$\EE$ is the sublinear part of $H(L)$.
This gives important general insight into the block entropy's behavior. It is
also quite practical, though: If $H(L)$ actually meets the asymptote at some
finite sequence length $R$, then the process is effectively an \order{R} Markov
chain \cite{Crut01a,Maho09a}: $\Prob(X_0 | \Past ) = \Prob(X_0|X_{-R}^R)$.
Interestingly, many finitary processes do not reach the asymptote at finite
lengths and so cannot be recast as Markov chains of any order. Roughly speaking,
they have various kinds of infinite-range correlation.


\vspace{-0.1in} 
\subsection{The Convergence Hierarchy}

In this way, the study of how the block entropy converges, or does not, is
a tool for
classifying processes. Reference~\cite{Crut01a} showed that the entropy rate and
excess entropy are merely two players in an infinite hierarchy that determines
the shape of $H(L)$. The central idea is to take $L$-derivatives
and integrals of $H(L)$.

To start, one has the block entropy difference:
\begin{align}
\Delta H[\MeasSymbol_0^L] \equiv H[\MeasSymbol_0^L] - H[\MeasSymbol_0^{L-1}] ~,
\end{align}
where $\Delta$ is the discrete derivative with respect to block length $L$.
It is easy to see that the right-hand side is the conditional entropy
$H[X_{L-1}|\MeasSymbol_0^{L-1}]$ and that, in turn,
\begin{align}
\hmu & = \lim_{L \rightarrow \infty} H[X_{L-1}|\MeasSymbol_0^{L-1}]  \\
     & = H[X_0 | \Past_0] ~,
\end{align}
recovering the entropy rate. It is often useful to directly refer
to the length-$L$ approximation to the entropy rate as
$\hmu(L) \equiv H[X_{L-1}|\MeasSymbol_0^{L-1}]$. $\hmu(L) \geq \hmu$ and
so it converges from above.

The excess entropy, for its part, controls the convergence speed, as it is the
discrete integral:
\begin{align}
\EE = \sum_{L = 1}^\infty (\hmu(L) - \hmu) ~.
\label{eq.ExcessEntropyPartialSum}
\end{align}
It requires only a few steps to see that this form is equivalent to that of
Eq.~(\ref{eq:ExcessEntropyMI}).

Following a similar strategy, the discrete integral
\begin{align}
\TI = \sum_{L = 0}^\infty \left[ \EE + \hmu L - H(L) \right]
\end{align}
measures how $H(L)$ itself reaches its linear asymptote $\EE + \hmu L$.
$\TI$ is called the \emph{transient information} and it is implicated in
determining the Markov order and, as we will show, synchronization.

The pattern should be clear now: At the lowest level, the transient information
indicates how quickly the block entropy reaches its asymptote. Then, that
asymptote grows at the rate $\hmu$ and has $y$-intercept $\EE$.
It might be helpful to refer to the graphical summary of block-entropy
convergence and the associated information measures given in Ref.
\cite[Fig. 2]{Crut01a}. Analogous diagrams will appear shortly.

All this can
be compactly summarized by introducing two operators: a derivative and an
integral that operate on $H(L)$. The derivative operator at the
$n^{\mathrm{th}}$-level is:
\begin{align}
\Delta^n H[\MeasSymbol_0^L]
    \equiv \Delta^{n-1} H[\MeasSymbol_0^L] - \Delta^{n-1}
	H[\MeasSymbol_0^{L-1}] ~,
\label{eq:DiffOperator}
\end{align}
for $L \geq n = 1, 2, \ldots$ and for $L \geq n=0$,
\begin{align}
\Delta^0 H[\MeasSymbol_0^L] \equiv H[\MeasSymbol_0^L] ~.
\end{align}
The integral operator is:
\begin{align}
  \mathcal{I}_n
    \equiv \sum_{L=n}^\infty \left[ \Delta^n H[\MeasSymbol_0^L] -
      \lim_{\ell\rightarrow \infty} \Delta^n H[\MeasSymbol_0^\ell] \right] ~,
\label{eq:IntegralOperator}
\end{align}
$n = 0, 1, 2, \ldots$.
(This is a slight deviation from Ref.~\cite{Crut01a}, when $n = 2$. See
App.~\ref{RURODiff}.)

To make the connection with what we just discussed, in this notation we have:
\begin{align}
  \hmu & = \lim_{L \rightarrow \infty} \Delta^1 H[\MeasSymbol_0^L] ~,\\
  \EE  & = \mathcal{I}_1 ~, ~\mathrm{and}\\
  \TI  & = - \mathcal{I}_0 ~.
\end{align}
Additionally, $\mathcal{I}_2$ is a process's \emph{total predictability}
$\mathbf{G}$
and $\Delta^2 H[\MeasSymbol_0^L]$ is its predictability gain---the rate at which
predictions improve by going to longer sequences.

The two operators, $\Delta_n$ and $\mathcal{I}_n$, define the
\emph{entropy convergence hierarchy} for a process, capturing those
properties reflected in the process's block entropy. Given a process's
specification, one attempts to calculate the hierarchy analytically; given
data, to estimate it empirically. In addition to systematizing a process's
informational properties, the hierarchy has a number of uses. For example,
structural classes of processes can be distinguished by the $n^*$ at which the
hierarchy becomes trivial; for example, when
$\Delta^n H[\MeasSymbol_0^L] = 0, ~n > n^*$. Other classifications turn on
bounded $\mathcal{I}_{n^*}$. The finitary processes, for example, are defined
by $n^* = 1$: $\mathcal{I}_1 = \EE < \infty$. Or, conversely, there are well
known processes for which some integrals diverge; they include the onset of
chaos through period-doubling, where the excess entropy diverges.
Reference \cite[Sec. VII.A]{Crut01a} introduces a classification of
processes along these lines.

\section{Process Presentations}

\subsection{The Causal State Representation}

Prediction is closely allied to the view of a process as a communication channel:
We wish to
predict the future using information from the past. At root, a prediction is
probabilistic, specified by a distribution of possible futures $\Future$ given
a particular past $\past$: $\Pr(\Future|\past)$. At a minimum, a good predictive
model needs to capture \emph{all} of the information $I$ shared between the
past and future: $\EE = I[\Past;\Future]$.

Consider now the goal of modeling---building a representation that allows not
only good prediction but also expresses the mechanisms producing a system's
behavior. To build a model of a structured process (a memoryful channel),
computational mechanics~\cite{Crut88a} introduced an equivalence relation
$\past \sim \past^\prime$ that clusters all histories which give rise to the
same prediction:
\begin{equation}
\epsilon(\past) =
  \{ \past^\prime: \Pr(\Future|\past) = \Pr(\Future|\past^\prime) \} ~.
\label{Eq:PredictiveEquivalence}
\end{equation}
In other words, for the purpose of forecasting the future, two different pasts
are equivalent if they result in the same prediction. The result of applying
this equivalence gives the process's \emph{causal states}
$\CausalStateSet = \Pr(\Past,\Future) / \sim$, which partition the space
$\AllPasts$ of pasts into sets that are predictively equivalent.
The set of causal states
\footnote{A process's causal states consist of both transient and recurrent
states. To simplify the presentation, we henceforth refer \emph{only} to
recurrent causal states.}
can be discrete, fractal, or continuous;
see, e.g., Figs. 7, 8, 10, and 17 in Ref.~\cite{Crut92c}.

State-to-state transitions are denoted by matrices
$T_{\CausalState \CausalState^\prime}^{(\meassymbol)}$ whose elements give the
probability $\Pr(\MeasSymbol=\meassymbol,\CausalState^\prime|\CausalState)$ of transitioning
from one state $\CausalState$ to the next $\CausalState^\prime$ on seeing
measurement $\meassymbol$. The resulting model, consisting of the causal
states and transitions, is called the process's \emph{\eM}. Given a process
$\mathcal{P}$, we denote its \eM\ by $M(\mathcal{P})$.

Causal states have a Markovian property that they render the past and future
statistically independent; they \emph{shield} the future from the past~\cite{Crut98d}:
\begin{equation}
\Pr(\Past,\Future|\CausalState)
  = \Pr(\Past|\CausalState) \Pr(\Future|\CausalState) ~.
\label{shield}
\end{equation}
Moreover, they are optimally predictive~\cite{Crut88a} in the sense that
knowing which causal state a process is in is just as good as having the
entire past: $\Pr(\Future|\CausalState) = \Pr(\Future|\Past)$. In other
words, causal shielding is equivalent to the fact~\cite{Crut98d} that the
causal states capture all of the information shared between past and future:
$I[\CausalState;\Future] = \EE$.

\EMs\ have an important structural property called
\emph{unifilarity}~\cite{Crut88a,Shal98a}: From the start state, each symbol
sequence corresponds to exactly one sequence of causal
states \footnote{In the theory of computation, unifilarity is
referred to as ``determinism'' \cite{Hopc79}.}. The importance
of unifiliarity, as a property of any model, is reflected in the fact that
representations without unifilarity, such as generic hidden Markov models,
\emph{cannot} be used to directly calculate important system
properties---including the most basic, such as how random a process is.
As a practical result, unifilarity is easy to verify: For each state, each
measurement symbol appears on at most one outgoing transition
\footnote{Specifically, each transition matrix $T^{(x)}$ has, at most, one
nonzero component in each row.}. Thus, the signature of unifilarity is that on
knowing the current state $\CausalState_t$ and measurement $\MeasSymbol_t$,
the uncertainty in the next state $\CausalState_{t+1}$
vanishes: $H[\CausalState_{t+1}|\CausalState_t,\MeasSymbol_t] = 0$.

Out of all optimally predictive models $\PrescientStateSet$---for which
$I[\PrescientState;\Future] = \EE$---the \eM\ captures the minimal
amount of information that a process must store in order to communicate all
of the excess entropy from the past to the future.  This is the Shannon
information contained in the causal states---the
\emph{statistical complexity}~\cite{Crut98d}:
$\Cmu \equiv H[\CausalState] \leq H[\PrescientState]$. It turns out that
statistical complexity upper bounds the excess entropy \cite{Crut88a,Shal98a}:
$\EE \leq \Cmu$. In short, $\EE$ is the
effective information transmission rate of the process, viewed
as a channel, and $\Cmu$ is the memory stored in that channel.

Combined, these properties mean that the \eM\ is the basis against which
modeling should be compared, since it captures all of a process's
information at maximum representational efficiency.

Importantly, due to unifilarity one can calculate the entropy rate directly
from a process's \eM:
\begin{align}
\hmu
  & = H[\MeasSymbol|\CausalState] \nonumber \\
  & = - \sum_{\{\CausalState\}} \Pr(\CausalState)
  \sum_{\{\meassymbol \CausalState^\prime\}} T^{(\meassymbol)}_{\CausalState \CausalState^\prime}
  \log_2 \sum_{\{\CausalState^\prime\}}{T^{(\meassymbol)}_{\CausalState\CausalState^\prime}}
  ~  .
\label{eq:hmuEM}
\end{align}
$\Pr(\CausalState)$ is the asymptotic
probability of
the causal states, which is obtained as the normalized principal eigenvector
of the transition matrix $T = \sum_{\{x\}} T^{(x)}$. A process's statistical
complexity can also be directly calculated from its \eM:
\begin{align}
\Cmu
   & = H[\CausalState] \nonumber \\
   & = - \sum_{\{\CausalState\}} \Pr(\CausalState) \log_2 \Pr(\CausalState) ~.
\end{align}
Thus, the \eM\  directly gives two important properties: a process's rate
($\hmu$) of producing information and the amount ($\Cmu$) of historical
information it stores in doing so. Moreover, Refs.~\cite{Crut08a,Crut08b}
showed how to calculate a process's excess entropy directly from the \eM.

\vspace{-0.1in} 
\subsection{General Presentations}

The \eM\ is only one possible description of a process. There are many
alternatives: Some larger, some smaller; some with the same prediction error,
some with larger prediction error; some that are unifilar, some not; some
that do an excellent job of capturing $\Prob(\Past,\Future)$, many (or most)
doing only an approximate job; some allowing for the direct calculation of
the process's properties, some precluding such calculations.

The \eM, compared to all other possible descriptions, is arguably the best.
The results in the following, as an ancillary benefit, strengthen this
conclusion considerably. However, it is important to keep in mind that due to
implementation constraints or intended use or under specified performance
criteria, alternative models may be desirable and preferred to the \eM.
Reference~\cite{Crut92c}, for example, compares the benefits and disadvantages
of different kinds of nonunifilar models that are smaller than the \eM.
We return to elaborate on this in Sec.~\ref{sec:Nonunifilar}.

One refers to a process's possible descriptions as \emph{presentations}
\cite{Lind95a}.
Specifically, these are state-based models---using states and state
transitions---that exactly describe $\Prob(\Past,\Future)$. That is, given a
finitary process $\Process$, we consider the set of all presentations that
generate the same \emph{process language}:
$\Prob(\MeasSymbol^L), ~ L = 1, 2, \ldots$.
The set of $\Process$'s presentations is the focus of our work here. That is,
we do not address models that give only approximations to the process language.

We refer to these alternative models as \emph{rivals}. A rival consists of
a set $\AlternateStateSet$ of states and state-to-state transitions
$T_{\AlternateState\AlternateStatePrime}^{(s)}$ over the symbols $s$ in the
process's measurement alphabet $\MeasAlphabet$. There is an associated mapping
$\eta: \past \rightarrow \AlternateStateSet$ from pasts to rival states. When
we refer to the rival's state as a random variable, we will denote this
$\AlternateState = \eta(\Past)$. We use lower case $\alternatestate$ when we
refer to a particular realization: $\AlternateState = \alternatestate$,
$\alternatestate \in \AlternateStateSet$. Just as with the \eM, given a rival
presentation, we can refer to the amount of information the rival states
contain---this is the \emph{presentation state entropy} $H[\AlternateState]$.

Above, we noted that a process's \eM\ is its minimal unifilar presentation.
But, how are the rivals related, if at all, to the \eM? To explore the
organization of the space of rivals, in the following we relax properties that
make the \eM\ unique, working with presentations that are nonminimal unifilar
and those that are not even unifilar. And so, we must distinguish several kinds
of presentation. First, we extend unifilarity to presentations, generally.
\begin{Def}
A presentation is \emph{unifilar} if and only if
$H[\AlternateState_{t+1} | \AlternateState_t, \MeasSymbol_t ] = 0$.
\end{Def}
\noindent
Second, we introduce the notion of reverse-time unifilarity.
\begin{Def}
A presentation is \emph{counifilar} if and only if
$H[\AlternateState_t | \MeasSymbol_t, \AlternateState_{t+1} ] = 0$.
\end{Def}
\noindent
Third, we will consider prescient presentations, those whose states are as good
at predicting as the \eM's causal states \cite{Crut98d,Shal98a}.
\begin{Def}
A presentation is \emph{prescient} if and only if, for all pasts
$\past \in \AllPasts$:
\begin{equation}
\Prob(\Future^L|\AlternateState = \eta(\past)) =
\Prob(\Future^L|\CausalState = \epsilon(\past) ) ~,
\end{equation}
for all $L \geq 1, 2, 3, \ldots$.
\end{Def}
\noindent
We will also shortly discuss presentations to which one can or cannot
synchronize---that are or are not controllable.

\vspace{-0.1in} 
\section{State-Block and Block-State Entropies}

Now, we introduce two block entropies and discuss their properties,
but first, we recall some well known results from information
theory~\cite[Sec. 4.2]{Cove06a}.

For any stationary stochastic process,
$\Delta \BE{L}$ is a nonincreasing sequence of nonnegative terms that
converges, from above, to the entropy rate $\hmu$.
There is a complementary result which provides an estimate of the entropy
rate that converges from below.  It is typically stated in terms of the
Moore (state-output) type of hidden Markov model~\cite[Thm. 4.5.1]{Cove06a},
so we recast the theorem in terms of the Mealy (edge-output) type of hidden
Markov models, used exclusively here.

\begin{The}
If $\AlternateState_0, \AlternateState_1, \ldots$ form a stationary Markov chain
and $(\MeasSymbol_i, \AlternateState_{i+1}) = \phi(\AlternateState_i)$, then
\begin{equation}
 H[\MeasSymbol_L | \AlternateState_0, \MeasSymbol_0^L]
  \leq \hmu \leq
 H[\MeasSymbol_L|\MeasSymbol_0^L] ~,
\end{equation}
$L = 0, 1, 2, \ldots$, and
\begin{equation}
H[\MeasSymbol_\infty | \AlternateState_0, \Future_0] = \hmu ~.
\end{equation}
$\phi$ need not be a deterministic mapping.
\label{sbeconvergence}
\end{The}
\noindent
Appendix~\ref{sbelowerboundhmu} provides the proof details.
Henceforth, we refer to $\AltSBE{L}$ as the \emph{state-block entropy}.

We also define the \emph{block-state} entropy to be $\AltBSE{L}$. As with
the state-block entropy, there is a corresponding convergence result.
\begin{The}
If $\AlternateState_0, \AlternateState_1, \ldots$ form a stationary Markov chain
and $(\MeasSymbol_i, \AlternateState_{i+1}) = \phi(\AlternateState_i)$, then
\begin{equation}
 \AltBSE{L} - \AltBSE{L-1}
  \leq \hmu \leq
 H[\MeasSymbol_L|\MeasSymbol_0^L] ~,
\end{equation}
$L = 1, 2, 3, \ldots$, and
\begin{equation}
\lim_{L\rightarrow \infty} \biggl( \AltBSE{L} - \AltBSE{L-1} \biggr) = \hmu ~.
\end{equation}
\label{bseconvergence}
Again, $\phi$ need not be a deterministic mapping.
\end{The}
Ref.~\cite{Maho10a} provides the proof of this theorem and discusses related
results in the context of crypticity and cryptic order \cite{Maho09a}.

Note, both of these theorems hold for general presentations---not
just \eMs---and this fact serves as the motivation for our later
generalizations.

\vspace{-0.1in} 
\subsection{Convergence Hierarchies}

Just as with the block entropy $\BE{L}$, we will consider $L$-derivatives
and integrals of the state-block and block-state entropies. At the first level,
\begin{align}
\Delta \AltSBE{L} &\equiv  \AltSBE{L} - \AltSBE{L-1} ~,\\
\Delta \AltBSE{L} &\equiv  \AltBSE{L} - \AltBSE{L-1} ~.
\end{align}
Higher-order derivatives are defined similarly to Eq.~\eqref{eq:DiffOperator}.
As before, the $n=0$ case is an identity operator. So, for example,
$\Delta^0 \AltSBE{L} = \AltSBE{L}$.

We already know---Thms. \ref{sbeconvergence} and \ref{bseconvergence}---that
both of these quantities tend to $\hmu$ in the large-$L$ limit,
ensuring that all higher-order derivatives tend to zero.

Now, consider the $n^\text{th}$ state-block and block-state integrals:
\begin{align}
\mathcal{K}_n &= \sum_{L=n}^\infty \bigl( \Delta^n \AltSBE{L}
  - \lim_{\ell\rightarrow\infty} \Delta^n \AltSBE{\ell} \bigr) ~,\\
\mathcal{J}_n &= \sum_{L=n}^\infty \bigl( \Delta^n \AltBSE{L}
  - \lim_{\ell\rightarrow\infty} \Delta^n \AltBSE{\ell} \bigr) ~.
\end{align}
Note that both $\mathcal{K}_0 \geq 0$ and $\mathcal{J}_0 \geq 0$ while, in
contrast, $\mathcal{I}_0 \leq 0$. Also, $\mathcal{K}_1 \leq 0$ and
$\mathcal{J}_1 \leq 0$ while $\mathcal{I}_1 \geq 0$. These differences are
due to the fact that the block entropy is concave in $L$ while the state-block
and block-state entropies are convex.

Consider the partial sums of $\mathcal{K}_1$---the state-block integral:
\begin{align}
  \mathcal{K}_1(L)
    &= \sum_{\ell=1}^L \bigl( \Delta \AltSBE{\ell} - \hmu \bigr)
	\nonumber \\
    &= \AltSBE{L} - \AltSBE{0} - L \hmu
	\nonumber \\
    &= H[\MeasSymbol_0^L|\AlternateState_0] - L \hmu ~.
\end{align}
Note that if the presentation is unifilar, then
$H[\MeasSymbol_0^L | \AlternateState_0] = L \hmu$ and $\mathcal{K}_1(L) = 0$.
Thus, unifilarity is a sufficient condition for $\mathcal{K}_1 = 0$, but it is
not a necessary condition.

Now, consider the partial sums of $\mathcal{J}_1$---the block-state integral:
\begin{align}
  \mathcal{J}_1(L)
    &= \sum_{\ell=1}^L \bigl( \Delta \AltBSE{\ell} - \hmu \bigr)
	\nonumber \\
    &= \AltBSE{L} - \AltBSE{0} - L \hmu
	\nonumber \\
    &= \AltBSE{L} - \AltBSE[L]{0} - L \hmu
	\nonumber \\
    &= H[\MeasSymbol_0^L|\AlternateState_L] - L \hmu ~.
\end{align}
Similarly, if the presentation is counifilar, then it follows that
$H[\MeasSymbol_0^L | \AlternateState_L] = 0$ and $\mathcal{J}_1(L) = 0$.
So, counifilarity is a sufficient condition for $\mathcal{J}_1 = 0$, but
it is not a necessary condition.

\vspace{-0.2in} 
\subsection{Asymptotics}

Theorems~\ref{sbeconvergence} and \ref{bseconvergence} tell us $\SBE{L}$ and
$\BSE{L}$ are convex functions in $L$ and that the slope limits to the
entropy rate. This means that each curve converges to a linear asymptote,
cf. Eq.~\eqref{eq:beasymptote}:
\begin{align}
\label{eq:sbeasymptote}
\AltSBE{L} &\propto Y_\text{SBE} + \hmu L \\
\label{eq:bseasymptote}
\AltBSE{L} &\propto Y_\text{BSE} + \hmu L ~,
\end{align}
where $Y_\text{SBE}$ and $Y_\text{BSE}$ are constants independent of $L$.
The pictures that one should have in mind for the growth of these entropies
are those of Figs. \ref{fig:SI}, \ref{fig:eMEntropyGrowth},
\ref{fig:ASEntropyGrowth}, \ref{fig:UnifEntropyGrowth}, and
\ref{fig:NUnifEntropyGrowth}, which we will discuss in due course.

In fact, we will take this behavior as the definition of the following
linear asymptotes:
\begin{gather}
\begin{aligned}
Y_\text{SBE} & \equiv \lim_{L\rightarrow\infty}
                      \biggl( \AltSBE{L} - \hmu L \biggr) \\
             & = \lim_{L\rightarrow\infty}
                      \biggl( H[\AlternateState_0]
                  + H[\MeasSymbol_0^L | \AlternateState_0] - \hmu L \biggr) \\
             & = H[\AlternateState_0] + \mathcal{K}_1
\end{aligned}
\label{eq:sbeasymptote2}
\end{gather}
and
\begin{gather}
\begin{aligned}
Y_\text{BSE} & \equiv \lim_{L\rightarrow\infty}
                      \biggl( \AltBSE{L} - \hmu L \biggr) \\
             & = \lim_{L\rightarrow\infty}
                      \biggl( H[\AlternateState_L]
                  + H[\MeasSymbol_0^L | \AlternateState_L] - \hmu L \biggr) \\
             & = H[\AlternateState_0] + \mathcal{J}_1 ~.
\end{aligned}
\label{eq:bseasymptote2}
\end{gather}
These tell us that $\mathcal{K}_1$ and $\mathcal{J}_1$ are not the
sublinear parts of the state-block and block-state entropies. This is in
contrast to the corresponding result for the block entropies:
\begin{align}
Y_\text{BE} &\equiv \lim_{L \rightarrow \infty}
                      \biggl (\BE{L} - \hmu L \biggr) \\
            &= \lim_{L \rightarrow \infty}
                      \biggl (\BE{0} + \BE{L} - \hmu L \biggr) \\
            &=\BE{0} + \mathcal{I}_1 ~.
\end{align}
The term $\BE{0}$ was dropped in the earlier partial sum formulation---i.e.,
Eq.~(\ref{eq.ExcessEntropyPartialSum})---since it corresponds to no measurement
being made and so is zero. It is reintroduced above, though, to complete the
formal parallel to the state-block and block-state entropy cases.

The result for block entropy is that the offset of the linear asymptote was
equal to the $\mathcal{I}_1$, the excess entropy. However, the argument just
given clearly establishes that, in fact, one should think of the first
derivatives as offsets from the initial value of their corresponding curves.

Finally, recall that $\mathcal{K}_1$ and $\mathcal{J}_1$ are not greater than
zero, so $Y_\text{SBE}$ and $Y_\text{BSE}$ are less than or equal to the
presentation state entropy $H[\AlternateState_0]$.

\vspace{-0.2in} 
\section{Synchronization}

\vspace{-0.1in} 
\subsection{Duality of Synchronization and Control}

Synchronization is a question about how an observer comes to know a process's
(typically hidden) current internal state through observations. (Recall the
picture introduced in Ref.~\cite{Crut01b}.) As such, it
requires a notion of state, either the process's causal state (using the
\eM) or the state of some other presentation. In either case we monitor the
observer's uncertainty over the states $\AlternateStateSet$ after having seen
a series of measurements $w = \meassymbol_0\meassymbol_2\ldots\meassymbol_{L-1}$
using the conditional state entropy $H[\AlternateState|w]$. When this
vanishes, the observer is synchronized and we call $w$ a
\emph{synchronizing word}.

During synchronization, the observer updates her answer to the question,
``Which presentation states can be reached by sequence $w$?'' When
there is a unique answer, the observer is synchronized. If the eventual answer,
though, is only a proper subset of presentation states, then
$0 < H[\AlternateState|w] \le H[\AlternateState]$ and the observer can be said
to be partially synchronized.

A formal treatment of synchronization appears in Refs.~\cite{Trav10a,Trav10b},
which define asymptotic synchronization as follows.

\begin{Def}
A presentation is \emph{weakly asymptotically synchronizing} if and only if
$\lim_{L \rightarrow \infty} H[\AlternateState_L | \MeasSymbol_{0}^L ] = 0$.
\end{Def}

While some processes can have synchronizing words, others have
\emph{synchronizing blocks} where every word of a finite length $R$ is a
synchronizing word. Such processes are called \emph{Markov processes}.
The smallest such $R$ is the \emph{Markov order} \cite{Maho09a,Jame10a}. It
turns out that the \eM\ presentation for a Markov process is
\emph{exactly synchronizing}~\cite{Trav10a}: for finite $R$,
$H[\CausalState_0|\MeasSymbol_0^L] = 0, L \geq R$.

If a process admits a presentation that is only weakly asymptotically
synchronizing, though, then an observer will be in various conditions of
state uncertainty until the limit $L \rightarrow \infty$. Finitary \eMs,
as it turns out, are always weakly asymptotically synchronizing and the
state uncertainty vanishes exponentially fast~\cite{Trav10b}:
$\Prob \left( H[\CausalState_0|\MeasSymbol_0^L] > 0 \right) \propto e^{-L}$.

The controllability properties of a process and its models are analogous.
However, now there is a designer that has built an implementation of a process.
And, starting from an unknown condition, the designer wishes to prepare the
process in a particular state or set of states by imposing a sequence of inputs.
Phrased this way, one sees that the implementation is, in effect, a presentation
and the control sequence is none other than a synchronizing word. Due to this
duality, we only discuss synchronization in the bulk of our development,
returning at the end to briefly draw out interpretations of the results for
controllability.

\vspace{-0.2in} 
\subsection{Synchronizing to the \EM}
\label{synceM}
\vspace{-0.1in} 

We noted that the \eM\ directly gives two important information-theoretic
properties---the entropy rate ($\hmu$) and the statistical complexity
($\Cmu$)---and one (the excess entropy $\EE$) indirectly. The difference between
$\Cmu$ and $\EE$ was introduced as the \emph{crypticity}~\cite{Crut08a,Crut08b}
\begin{align}
\PC = \Cmu - \EE
\label{eq:crypticity}
\end{align}
to describe how much of the internal state information ($\Cmu$) is not locally
present in observed sequences ($\EE$).

Synchronization, as we discussed, is a property of the recurrent portion of the \eM\ and since it
is unifilar, if one knows its current state and follows transitions according
to the word being considered, then one will always know the \eM's final state.
However, it is also useful to consider the scenario when one does not know the
\eM's current state. Given no other information, the best estimate for the
current state is to draw from the stationary state distribution
$\Prob(\CausalState)$. Then, as each symbol is observed, one updates this
\emph{belief} distribution and estimates the next state from this updated
distribution.

As noted above, $H[\CausalState_L | \MeasSymbol_0^L]$ converges to zero
exponentially fast for all \eMs\ with a finite number of recurrent causal
states. At each $L$ before that point, there is an uncertainty in the causal
state of the \eM. If we add up the uncertainty at each length, then we have
the \emph{synchronization information}:
\begin{align}
\label{eq:synchronization}
 \SI & \equiv \sum_{L=0}^\infty H[\CausalState_L | \MeasSymbol_0^L]\\
\label{eq:sibse}
     & = \sum_{L=0}^\infty \biggl( \BSE{L} - \BE{L} \biggr) ~.
\end{align}
Importantly, the second line shows that synchronization information can be
visualized as the sum of all differences between the block-state and the
block entropy curves. Moreover, starting from Eq.~(\ref{eq:sibse}) we find:
\begin{align}
\label{eq:sipmlhmu}
 \SI & = \sum_{L=0}^\infty \biggl( \BSE{L} - (\EE + L\hmu) \biggr) \nonumber \\
     &\qquad -\sum_{L=0}^\infty \biggl( \BE{L} - (\EE + L\hmu) \biggr)\\
\label{eq:sij0i0}
     & = \mathcal{J}_0 - \mathcal{I}_0 ~.
\end{align}
We know that $\TI = -\mathcal{I}_0$. When we identify $\mathcal{J}_0$ with a
separate, nonnegative information quantity we conclude immediately that
$\SI \geq \TI$. This relationship is shown graphically in Fig.~\ref{fig:SI}.

\begin{figure}
\centering
\ifcolor
    \includegraphics[width=\columnwidth]{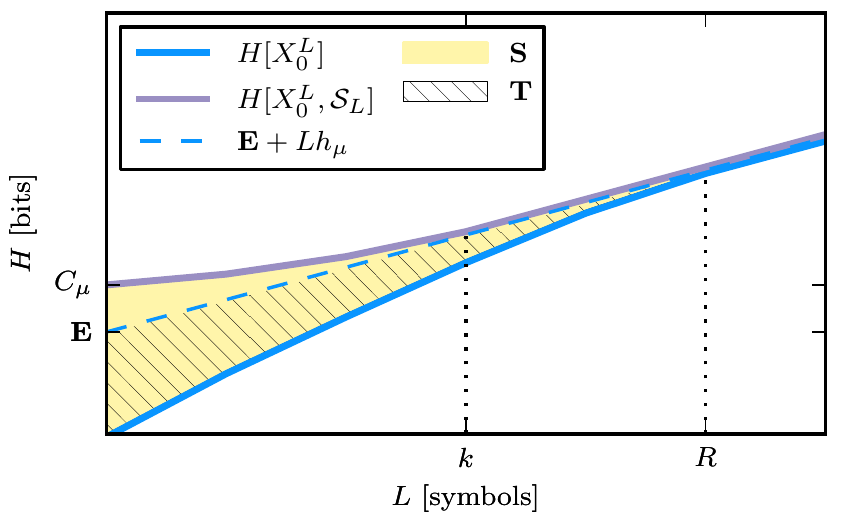}
\else
    \includegraphics[width=\columnwidth]{images/SvsT_gray}
\fi
\caption{Block entropy and block-state entropy growth for a generic finitary
  stationary process: It is easily seen that the synchronization information
  upper bounds the transient information, $\TI \leq \SI$, as $\TI$ is a
  component of $\SI$. The Markov order $R$ and cryptic order $k$ are also
  shown in their proper relationship $k \leq R$: $R$ indicates where the block
  entropy meets the $\EE + \hmu L$ asymptote and $k$, where the block-state
  entropy meets the same asymptote.
  }
\label{fig:SI}
\end{figure}

The \emph{cryptic order} $k$, as defined in Ref.~\cite{Maho09a}, can be
interpreted as the length at which the block-state curve has converged to its
asymptote:
$\EE + \hmu L$. Surprisingly, this is not the length at which an \eM\
can be considered synchronized, which is given by the Markov order $R$.
Given its definition as
the smallest value $L$ for which $H[\CausalState_L | \Future_0] = 0$, we see
that the cryptic order can be interpreted as a measure of how far back in time
the state sequence can be retrodicted from the distant future.

For example, the Even Process consists of all bi-infinite sequences that
contain even-length stretches of $1$s separated by at least a single $0$;
see Ref.~\cite{Crut01a}.
This process cannot be considered synchronized at any finite length because all
the thus-far seen symbols may be $1$s, and so one does not know if the latest
symbol is a $1$ at an even- or odd-valued location. In contrast, once a $0$ has
been seen, we know instantly the evenness and oddness of each preceding $1$,
making the cryptic order $k = 0$. Since the cryptic order $k = 0$ for the Even
Process, one concludes that $\mathcal{J}_0$ does not contribute to $\SI$ and
$\TI = \SI$.

The two pieces---$\mathcal{J}_0$ and $\mathcal{I}_0$---comprising $\SI$ are both
finite due to the exponentially fast convergence of the two block-entropy
curves~\cite{Trav10b}. This shows that $\SI$ consists of distinct information
contributions drawn from different process features. Referring to
Fig.~\ref{fig:SI}, the lower piece, the
transient information $\TI$, is information recorded due to an over-estimation
of the entropy rate $\hmu$ at block lengths $L$ less than the Markov order $R$.
This over-estimation is due, in effect, to $L$ being shorter than the longest
correlations in the data. In a complementary way, the upper portion
$\mathcal{J}_0$ can be viewed as the amount of state information that cannot
be retrodicted, even given the infinite future.

The relative roles of the contributions to synchronization information can be
clearly seen for one-dimensional range-$\MOrder$ spin systems.
Reference~\cite{Crut01a} claimed that for spin chains:
\begin{align}
  \SI = \TI + \frac{1}{2}\MOrder(\MOrder+1)\hmu ~,
\end{align}
where $\MOrder$ is the coupling range (Markov order) of the spin chain. This
can be established rather directly, and understood for the first time, using
the geometric convergence picture just introduced for $\SI$. First,
Ref.~\cite{Maho10a} showed that for a spin chain \BSE{L} is flat (zero slope)
for $0 \leq L \leq R$, after which it converges to its asymptote. Second,
combining these, we have:
\begin{align}
  \mathcal{J}_0 & = \sum_{L=0}^\MOrder \BSE{L} - (\EE + L\hmu) \\
                & = \sum_{L=0}^\MOrder (\EE + R\hmu) - (\EE + L\hmu) \\
                & = \sum_{L=0}^\MOrder (\MOrder-L)\hmu \\
                & = \frac{1}{2}\MOrder(\MOrder+1)\hmu ~.
\end{align}
So, the amount of state information that cannot be retrodicted is quadratic in
Markov order.

Finally, \BE{L} and \BSE{L} give lower and upper bounds on \EE, respectively:
the first monotonically approaches $\EE + L\hmu$ from below and the second
monotonically approaches it from above. This way, given an \eM, it is simple
to compute $\EE$ with any accuracy required from the block
entropies, which themselves can be efficiently estimated from the \eM.
Similarly, since \BE{L} over-estimates the entropy rate while approaching from
above and \BSE{L} under-estimates the entropy rate while approaching from below,
one obtains an analogous pair of bounds on $\hmu$. This block-state technique
for bounding the entropy rate, however, holds for any type of presentation of
the process. (Cf. Ref.~\cite[Sec. 4.5]{Cove06a}.)

\vspace{-0.2in} 
\section{Presentation Quantifiers}

The development and results have focused, so far, on \eMs\ and their
information-theoretic properties. Due to the \eM's uniqueness, these were also
properties of the corresponding processes themselves. Now, we relax the defining
characteristics of \eMs\ to consider generic presentations. Naturally, this
destroys our ability to directly identify presentation properties with those of
the process represented. A process's entropy rate ($\hmu$) and excess entropy
($\EE$) remain unchanged, however, since they are defined solely through its observables
$\Prob(\Past, \Future)$. Widening our purview to generic presentations leads us to
briefly introduce several new properties that capture information processing
in presentations. Perhaps more distinctly, this also leads us to quantify
the kinds of information in a presentation that are \emph{not} characteristics
of the process it represents. Section~\ref{sec:classifying} then provides more
detailed expositions on their meaning and example processes to illustrate them.

\vspace{-0.2in} 
\subsection{Crypticity}

The statistical complexity $\Cmu$ is the amount of information a process must
store in order to generate future behavior. The crypticity $\PC$ is that part of
$\Cmu$ not transmitted to the future: $\PC = \Cmu - \EE$. Roughly, it can be
thought of as the irreducible overhead that arises from the process's causal
structure.  Reference~\cite{Crut08a} defined crypticity for \eMs\ as
$\PC = H[\CausalState_0 | \Future_0]$.  Now, we generalize this to define
crypticity for generic presentations.

\begin{Def}
The \emph{presentation crypticity} $\PC(\AlternateState)$ is the amount of
state information shared with the past that is not transmitted to the future:
\begin{align}
\PC \equiv I[\Past_0; \AlternateState_0 | \Future_0] ~.
\end{align}
\end{Def}

When the presentation states are causal states, this quantity reduces to the
original definition---the process's crypticity. Furthermore, the crypticity
is the difference between the presentation state entropy and the
$y$-intercept of block-state entropy curve, Eq.~\eqref{eq:bseasymptote}.

\begin{The}
The presentation crypticity $\PC(\AlternateState)$ is the difference between
the presentation state entropy $H[\AlternateState_0]$ and the
sublinear part of the block-state entropy:
\begin{align}
\PC = -\mathcal{J}_1 ~.
\end{align}
\end{The}
\begin{proof}
Starting with the length-$L$ approximation of the crypticity, we work our way
to the $L^\textrm{th}$ partial sum of $-\mathcal{J}_1$ via a straightforward
calculation:
\begin{align}
  & I[\MeasSymbol_{-L}^L; \AlternateState_0^{\vphantom{L}}| \MeasSymbol_0^L]
    \nonumber \\
  &\quad = H[\MeasSymbol_{-L}^L | \MeasSymbol_0^L]
    - H[\MeasSymbol_{-L}^L| \AlternateState_0^{\vphantom{L}},
                            \MeasSymbol_0^L] \\
  &\quad = H[\MeasSymbol_{-L}^L | \MeasSymbol_0^L]
    - H[\MeasSymbol_{-L}^L | \AlternateState_0^{\vphantom{L}}]
    \label{eq:pcproof-markov} \\
  &\quad = L \hmu - H[\MeasSymbol_{-L}^L | \AlternateState_0^{\vphantom{L}}]
    + H[\MeasSymbol_{-L}^L | \MeasSymbol_0^L] - L \hmu \\
  &\quad = -\mathcal{J}_1(L) +
    H[\MeasSymbol_{-L}^L | \MeasSymbol_0^L] - L \hmu
    \label{eq:pcproof-J1} \\
  &\quad = -\mathcal{J}_1(L) +
    H[\MeasSymbol_0^L | \MeasSymbol_{-L}^L] - L \hmu
    \label{eq:pcproof-stationarity1} \\
  &\quad = -\mathcal{J}_1(L) +
    \sum_{j=0}^{L-1} H[\MeasSymbol_j | \MeasSymbol_0^j, \MeasSymbol_{-L}^L]
    - L \hmu \label{eq:pcproof-chainrule} \\
  &\quad = -\mathcal{J}_1(L) +
   \sum_{j=L}^{2L-1} H[\MeasSymbol_j | \MeasSymbol_0^{L+j}] - L\hmu
    \label{eq:pcproof-stationarity2} ~.
\end{align}
Equation~\eqref{eq:pcproof-markov} follows because the states (in any hidden
Markov model) shield the past from the future: the future is a function of the
state. Equation~\eqref{eq:pcproof-J1} follows from the
definition of $\mathcal{J}_1$, and Eq.~\eqref{eq:pcproof-stationarity1}
from stationarity. Equation~\eqref{eq:pcproof-chainrule} follows
from the chain rule for block entropies~\cite{Cove06a}, and
Eq.~\eqref{eq:pcproof-stationarity2} from using stationarity again.

Finally, we take the large-$L$ limit. By definition,
we have $\mathcal{J}_1(L) \rightarrow \mathcal{J}_1$.  The remaining difference
converges to zero due to a result in Ref.~\cite{Trav10b} that the conditional
block entropies converge to the entropy rate faster than linearly in $L$.
\end{proof}

\vspace{-0.2in} 
\subsection{Oracular Information}

We now introduce a sibling of crypticity---the \emph{oracular
information}.

\begin{Def}
The \emph{oracular information} is the amount of state information
shared with the future that is not derived from the past:
\begin{align}
\OI \equiv I[\AlternateState_0; \Future_0 | \Past_0] ~.
\end{align}
\end{Def}

This new quantity is always zero for the \eM\ and nonzero only for
nonunifilar presentations. It is the difference between presentation statistical
complexity and the $y$-intercept of the state-block entropy curve,
Eq.~\eqref{eq:sbeasymptote}.

\begin{The}
The oracular information is the difference between the
presentation state entropy $H[\AlternateState_0]$ and the
sublinear part of the state-block entropy curve:
\begin{align}
\OI = -\mathcal{K}_1 ~.
\end{align}
\end{The}
\begin{proof}
The proof proceeds almost identically to the corresponding result for
crypticity. Namely,
\begin{align}
I[\AlternateState_0^{\vphantom{L}}; \MeasSymbol_0^L | \MeasSymbol_{-L}^L]
  = & -\mathcal{K}_1(L) \nonumber \\
  & + \sum_{j=L}^{2L-1} H[\MeasSymbol_j | \MeasSymbol_{0}^{j}] - L\hmu ~.
\end{align}
Then, taking the large-$L$ limit proves the result.
\end{proof}
In this sense, a positive oracular information indicates that there is a
deficit in using only the rival states for prediction. More information---the
oracular information---must be extracted from the presentation in order to
perform optimal prediction.

\vspace{-0.2in} 
\subsection{Gauge Information}

When moving away from the optimal representation afforded by a process's \eM,
it is possible to encounter presentations containing state information that is
not justified by a process's bi-infinite set of observables. We call this
\emph{gauge information} to draw a parallel with the descriptional degrees of
freedom that gauge theory addresses in physical systems~\cite{Fram08a}.

\begin{Def}
The \emph{gauge information} is the uncertainty in the presentation states
given the entire past and future:
\begin{align}
 \GI \equiv H[\AlternateState_0 | \Past_0, \Future_0] ~.
\end{align}
\end{Def}

That is, to the extent there is uncertainty in the states, even after the past
and the future are known, the presentation contains state uncertainty above and
beyond the process. Thus, there are components of the model that are not
determined by the process; rather they are the result of a choice of
presentation.

Intuitively, gauge information can be related to the total state entropy,
crypticity, oracular information, and excess entropy.  Later, we will
discuss information diagrams as a useful visualization tool, but for now,
we simply point out that one can ``visually'' verify the following theorem
from Figure~\ref{fig:NUnif}.

\begin{The}
Gauge information is the difference between the state entropy and the sum
of the crypticity, oracular information, and excess entropy:
\begin{equation}
  \GI = H[\AlternateState] - \left( \PC + \OI + \EE \right) ~.
\end{equation}
\end{The}
\begin{proof}
Since we are working with hidden Markov models, the future and past are
conditionally independent given the current state. Thus,
$\EE \equiv I[\Past; \Future] = I[\Past; \AlternateState; \Future]$.
Now, the proof proceeds as a simple verification:
\begin{align*}
\PC(L) + \OI(L) + \EE(L)
  &= I[\PastL; \AlternateState | \FutureL] \\
  &\quad+ I[\AlternateState; \FutureL | \PastL] \\
  &\quad+ I[\PastL; \AlternateState; \FutureL] \\
  &= H[\AlternateState] - H[\AlternateState | \PastL, \FutureL] ~.
\end{align*}
So, finite-length approximations to the gauge information can be written as:
\begin{align*}
H[\AlternateState] - \bigl( \PC(L) + \OI(L) + \EE(L) \bigr)
	= H[\AlternateState | \PastL, \FutureL] ~.
\end{align*}
Taking the limit, we achieve our desired result.
\end{proof}

\vspace{-0.2in} 
\subsection{Synchronization Information}

As we noted, it is always possible to asymptotically synchronize to an \eM\ with
a finite number of recurrent causal states. For some processes, synchronization
can happen in finite time. While in others, it can only happen in the limit as
the observation window tends to infinity. In either case, it is always true
that $H[\CausalState_\infty | \Future] = 0$.

When we generalize to presentations that differ from \eMs, it is no longer true
that one always synchronizes to the presentation states. In such cases, there
is irreducible state uncertainty, even after observing an infinite number of
symbols. This kind of state uncertainty cannot be reduced by past observations alone.
Due to this, the synchronization information, as previously defined, diverges.

\begin{Def}
The \emph{presentation synchronization information} is the total uncertainty
in the presentation states:
\begin{align}
\SI \equiv \sum_{L=0}^\infty H[\AlternateState_L^{\vphantom{L}} | \MeasSymbol_0^L] ~.
\end{align}
\end{Def}
We will show in Sec.~\ref{synchronizationorder} that this can be understood in
terms of the gauge and oracular informations.

\vspace{-0.2in} 
\subsection{Cryptic Order}

The cryptic order was defined in Ref.~\cite{Maho09a} as the minimum length
$k$ for which $H[\CausalState_k | \Future_0] = 0$. Reference~\cite{Jame10a}
shows that the cryptic order is a topological property of the
\emph{irreducible sofic shift} \cite{Lind95a} describing the support of the \eM.
However, we can understand the cryptic order geometrically as the length
$\COrder$ at which the block-state entropy $\BSE{L}$ reaches its asymptote; see
Eq.~\eqref{eq:sbeasymptote}. It turns out that this concept generalizes
directly to generic presentations.
\begin{Def}
The \emph{presentation cryptic order} is the length $k$ at which the
block-state entropy curve reaches its asymptote:
\begin{align}
 \COrder &\equiv \min \left\{ L : \AltBSE{L}
          = H[\AlternateState_0] - \PC + \hmu L \right\} ~.
\end{align}
\end{Def}
One would like to understand the cryptic order in terms of an explicit limit,
as done for \eMs, where cryptic order is the minimum $k$ for which
$H[\CausalState_k | \Future_0] = 0$. The obvious complication for presentations,
in general, is that one might never synchronize to a particular state. However,
it turns out that one can understand the presentation cryptic order in terms
of one's uncertainty in the distribution over
\emph{distributions of states}---that is, the uncertainty in the distribution
over \emph{mixed states} \cite{Uppe97a,Crut08b}.
Specifically, we frame the generalized cryptic order in terms of synchronizing
to distributions over presentation states. We outline the approach briefly;
a detailed exposition will appear elsewhere~\cite{Jame10a}.

As measurements are made, an observer's uncertainty in the state of the
presentation varies. However, the pattern of variation becomes regular
as more observations are made. The cryptic order, then, is understood as the
number of distributions over presentation states that one cannot know with
certainty from time $t=0$ given the entire future.  Said differently, the
cryptic order is the time at which an observer becomes absolutely
certain about the uncertainty in the presentation states.

\vspace{-0.2in} 
\subsection{Oracular Order}

The oracular order definition parallels those of the cryptic and the Markov
orders.

\begin{Def}
The \emph{oracular order} is the length $\OOrder$ at which the state-block
entropy curve reaches its asymptote:
\begin{align}
 \OOrder &\equiv \min \left\{ L :
      \AltSBE{L} = H[\AlternateState_0] - \OI + \hmu L \right\} ~.
\end{align}
\end{Def}

It always vanishes for \eMs. So, this new length scale is a property of the
presentation only and not of the process generated by the presentation.

\vspace{-0.2in} 
\subsection{Gauge Order}

The gauge order definition also parallels those of the cryptic, Markov, and
oracular orders.

\begin{Def}
The \emph{gauge order} is the length $\GOrder$ at which
$H[\AlternateState_0 | \MeasSymbol_{-L}^L \MeasSymbol_0^L]$ reaches its
asymptote.
\begin{align}
\GOrder \equiv \min \{ L :
	H[\AlternateState_0 | \MeasSymbol_{-L}^L, \MeasSymbol_0^L]  = \GI \} ~.
\end{align}
\end{Def}
Geometrically, we visualize the gauge order as the length at which the
difference between two
curves---$H[\MeasSymbol_{-L}^L, \AlternateState_0 , \MeasSymbol_0^L]$ and
$H[\MeasSymbol_{-L}^L, \MeasSymbol_0^L]$---becomes fixed to their
asymptotic difference.
\begin{The}
The \emph{gauge order} is the maximum of the Markov, cryptic, and oracular orders:
\begin{align}
\GOrder = \max \{ \MOrder, \COrder, \OOrder \} ~.
\end{align}
\end{The}
\begin{proof}
The gauge information can be understood as the left-over state information
after the excess entropy, crypticity, and oracular information~\footnote{
Oracular information cannot be extracted from the past observables. This point
will be discussed further in Sec.~\ref{sec:classifying}.} have been extracted:
\begin{align}
\GI = H[\AlternateState_0] - \EE - \PC - \OI ~.
\end{align}
Thus, as soon as the observer reaches each of the Markov, cryptic, and oracular
orders, the remaining state information exactly equals the gauge information.
\end{proof}

It is important to note that, unlike the Markov, cryptic, and oracular orders,
the gauge order does \emph{not} indicate a scale at which an amount of
information is contained.  Rather, it is more the opposite. The gauge order is
the length scale beyond which there is no point attempting to extract any more
state information (even with an oracle), precisely because this remainder is
the gauge information and, therefore, not correlated with the process language.
It corresponds to what in physics one calls a gauge freedom.

\vspace{-0.2in}
\subsection{Synchronization Order}
\label{synchronizationorder}

As mentioned in Sec.~\ref{synceM}, the length at which an observer has
synchronized to an \eM\ is always $\MOrder$, the Markov order.  Recall, any
\order{R} Markov process has
$I[\Future_\MOrder ; \Past_0 | \MeasSymbol_0^\MOrder] = 0$.
Synchronization to the \eM\ requires that
$H[\CausalState_L | \MeasSymbol_0^L] = 0$, and it is
straightforward to see that this holds for $L=R$. As we generalize to
non-\eM\ presentations, though, we must look beyond Markov order to
address the fact that one might only synchronize to distributions
over presentation states.
\begin{Def}
The \emph{presentation synchronization order} is the length $\SOrder$ at which
$H[\AlternateState_L | \MeasSymbol_0^L]$ reaches its asymptote:
\begin{align}
  \SOrder \equiv \min \{ L : H[\AlternateState_L |
  \MeasSymbol_0^L] = \GI + \OI \} ~.
\end{align}
\end{Def}
The motivation for this definition is that the asymptote is simply the
difference of the asymptotes for the block-state and block entropy curves. That
is, the synchronization order is also thought of as the length at which
the state uncertainty equals its irreducible state uncertainty:
$\GI + \OI = H[\AlternateState_0 | \Past_0]$.

Now, we show that the synchronization order must occur at either
the presentation cryptic order or the Markov order.
\begin{The}
The presentation synchronization order is the maximum of the Markov and
presentation cryptic orders:
\begin{align}
\SOrder = \max \{\MOrder, \COrder \} ~.
\end{align}
\end{The}
\begin{proof}
When both the block-state and block entropy curves have reached their asymptotes
the observer will have extracted $\EE+\PC$ bits of state information. This
leaves $H[\AlternateState_0] - \EE - \PC = \GI + \OI$ bits. This is exactly
the irreducible state uncertainty---that which cannot be learned from
the observables.
\end{proof}
Note that for \eMs: $\EE + \PC = \Cmu$. So, when an observer has extracted all
that can be learned about the process from the past observables, the observer
has learned everything about the causal states.

When the synchronization order is finite,
$H[\AlternateState_L | \MeasSymbol_0^L]$
is fixed at the presentation's irreducible state uncertainty for all
$L>\SOrder$. Then, it can be helpful to view the presentation
synchronization information as consisting of two contributions:
\begin{align}
 \SI = \sum_{L=0}^{\SOrder-1} H[\AlternateState_L | \MeasSymbol_0^L]
+ \sum_{L=\SOrder}^\infty (\GI + \OI) ~.
\end{align}
When the synchronization order is not finite, it can be useful to interpret
the synchronization information in a slightly different manner:
\begin{align}
 \SI = \mathcal{I}_0 + \mathcal{J}_0 + \lim_{L \rightarrow \infty} (\GI + \OI)L ~.
\label{eq:SyncInfoWhenInfinite}
\end{align}

\vspace{-0.2in}
\subsection{Synchronization Time}

Reference~\cite{Feld02a} defined the \emph{synchronization time} $\tau$ of a
periodic process to be the average time needed to synchronize to the states.
Let $w=w_0\cdots w_{p-1}$ be a cyclic permutation of the word that is repeated
by a periodic process having period $p$. It follows that
\begin{align}
	\Pr(\MeasSymbol_0^p = w) = \frac{1}{p} ~,
\end{align}
since any cyclic permutation is just as likely as another.  Now, while each
permutation has the same probability, it is not true that each permutation
is equally informative in terms of synchronization.  For example, consider
the process that repeats the word $00011$, indefinitely. If an observer
saw $01$, then the observer would be synchronized. In contrast, the observer
would not be synchronized if $00$ had been observed instead.
Reference~\cite{Feld02a} defined $\tau_w$ as the synchronization time of the
cyclic permutations of $w$. Then,
\begin{align}
	\tau = \sum_{w} \tau_w \Prob(\MeasSymbol_0^p=w) ~.
\end{align}
Since $\hmu = 0$ for all periodic processes,
\begin{align}
\Pr(\MeasSymbol_0^p=w) = \Prob(\MeasSymbol_0^{\tau_w}=w_0 \cdots w_{\tau_{w}-1}) ~.
\end{align}
Thus, we can rewrite $\tau$ suggestively as:
\begin{align}
	\tau = \sum_{w} \tau_w  \Prob(\MeasSymbol_0^{\tau_w}=w_0 \cdots w_{\tau_{w}-1}) ~.
\end{align}
Then, instead of summing over all cyclic permutations of $w$, we can just sum
over the set $\mathcal{L}_\text{sync}$ of all minimal synchronizing words. (A
word is a \emph{minimal synchronizing word} if no prefix of the word is also
synchronizing.) Now, we can extend $\tau$ to all finitary processes, not just
periodic ones.
\begin{Def}
The \emph{process synchronization time} is the average time required
to synchronize to the \eM's recurrent causal states:
\begin{align}
	\tau \equiv \sum_{w \in \mathcal{L}_\mathrm{sync}}
				|w| \Prob \left( \MeasSymbol_0^{|w|} = w \right) ~.
\end{align}
\end{Def}
Note that any \order{R} Markov process has $\tau \leq R$. The synchronization
time gives an intuition for how long it takes to synchronize to a stochastic
process.

As an example, recall the Even Process \cite{Crut01a}. It has the property
that there are arbitrarily long minimal synchronizing words. For example,
$1^k0$ is always a minimal synchronizing word, for any $k$. Despite this fact,
the synchronization time of the Even Process is $\tau = 10/3$. After repeatedly
observing sequences four symbols in length, on average an observer will be
synchronized to the states of the \eM.

When considering more general presentations it is not always the case that one
can synchronize to the states, as $\tau$ can be infinite. Just as with the
cryptic order, however, one can synchronize to distributions over the
presentation states.  This motivates the presentation synchronization time.
\begin{Def}
The \emph{presentation synchronization time} is the average time required
to synchronize to a recurrent distribution over presentation states.
\end{Def}
We provide an intuitive definition here, leaving a more detailed
discussion, where notation is properly developed, for a sequel.

\vspace{-0.2in}
\section{Classifying Presentations}
\label{sec:classifying}

The \eM\ is frequently the preferred presentation of a process, especially
when one is interested in understanding fundamental properties of the process
itself. However, one might be interested in the properties of particular
presentations of a process, and it would be helpful if there was an
analogous theory similar to that which has been developed for \eMs.

To develop this, we establish a classification of a process's presentations.
The classes are defined in terms of whether a presentation is nonunifilar,
unifilar, weakly asymptotically synchronizable, and minimal unifilar. The
result is shown in Fig.~\ref{fig:bullseye}, which shows that the presentation
classes form a nested hierarchy.

\begin{figure}
\centering
\includegraphics[scale=.95]{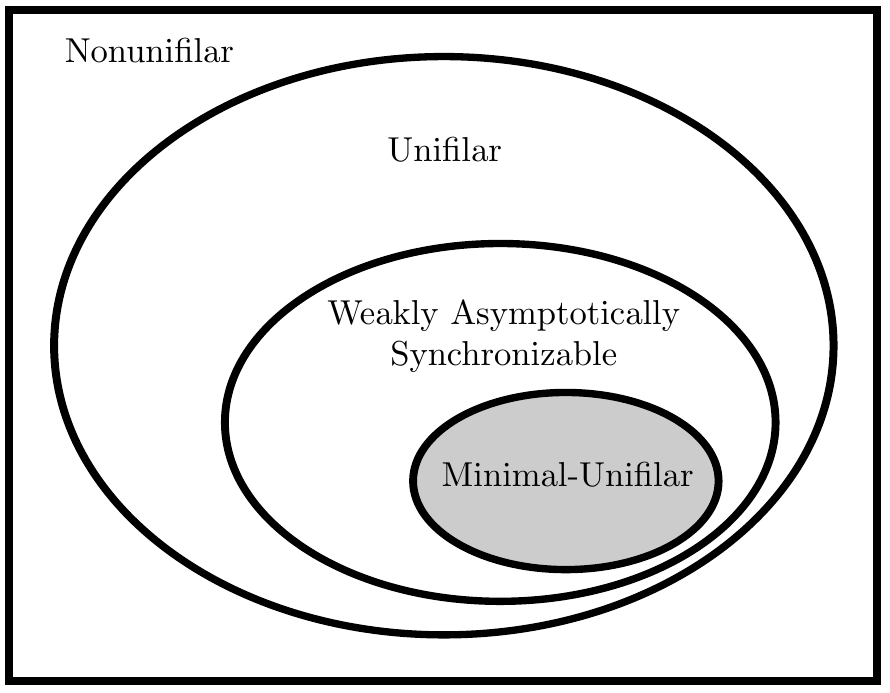}
\caption{The hierarchy of presentations of a finitary process. The gray region
  represents that portion to which the \eM\ belongs.
  }
\label{fig:bullseye}
\end{figure}

The most general type of presentation is nonunifilar, where we allow
the possibility that
$H[\AlternateState_1 | \AlternateState_0, \MeasSymbol_0] > 0$. Then, unifilar
presentations are the subset of nonunifilar presentations for which this
quantity is exactly zero. In the unifilar class, there can be redundant
states---states from which the future looks exactly the same and also states
which have the same exact histories mapping to them. When we move to weakly
asymptotically synchronizable presentations, all redundant states are removed
and the remaining states must induce a partition on the set of histories that
is a refinement of the causal state partition; cf. Ref. \cite[Lemma 7]{Shal98a}.
Finally, minimal unifilar presentations are the \eMs, whose partition of the
pasts is the coarsest one possible.

In this light, one might conclude that \eMs\ are an overly restricted set of
presentations. They are indeed a restricted set, but it is a restriction with
purpose: The \eM\ is the unique minimal prescient presentation within the set
of a process's presentations. Moreover, all of a process's properties can be
determined from its \eM. These facts allow one to purposefully conflate
properties of the \eM\ with process's properties.

We will use a \emph{information diagram} (I-diagram) \cite{Yeun91a} to analyze
what happens as one relaxes the defining properties of the \eM\ presentation's
random variables. With the \eM, we have the past $\Past$, the causal states
$\CausalState$, and the future $\Future$. As we move away from the \eM's causal
states, we must consider in addition the rival states $\AlternateState$.

\begin{figure*}
\centering
\ifcolor
    \includegraphics[scale=.65]{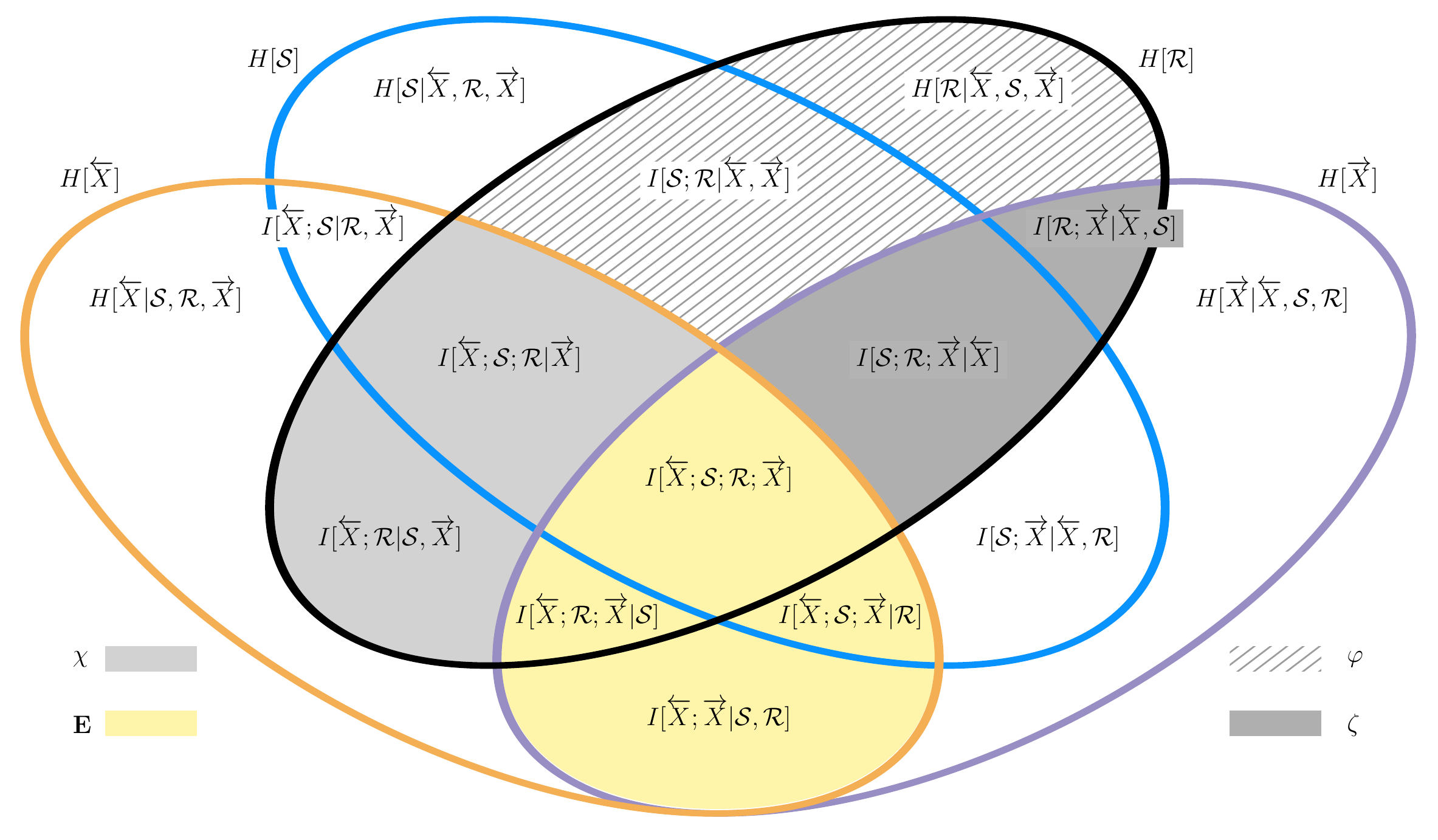}
\else
    \includegraphics[scale=.65]{images/4variable_gray}
\fi
\caption{The general four-variable information diagram involving
  $\protect\Past$, $\CausalState$, $\AlternateState$, and $\protect\Future$.
  The shaded light gray is the generalized crypticity $\PC$.
  The yellow is the excess entropy $\EE$.
  The dark gray is the oracular information $\OI$.
  The hatched area is the gauge information $\GI$.
  Note that this is only a schematic diagram of the interrelationships. In
  particular, potentially infinite quantities---such as, $H[\protect\Past]$
  and $H[\protect\Future]$---are depicted with finite areas.
  }
\label{fig:4variable}
\end{figure*}

In total, there are four random variables to consider. The full range of their
possible information-theoretic relationships appears in the information diagram
(I-diagram) of Fig.~\ref{fig:4variable}. However, Appendix~\ref{zeroatoms} shows
that 7 of the 15 atoms (elemental components of the multivariate information
measure sigma algebra) vanish. This allows us to simplify other atoms
dramatically. For example, the atom:
\begin{align}
&I[\Past; \CausalState; \AlternateState; \Future] \\
              &\quad = I[\Past; \CausalState; \Future]
                 - I[\Past; \CausalState; \Future | \AlternateState] \\
              &\quad = I[\Past; \CausalState; \Future] \\
              &\quad = I[\Past; \Future] - I[\Past; \Future | \CausalState] \\
              &\quad = I[\Past; \Future] -
                 (I[\Past; \AlternateState; \Future | \CausalState]
                 +I[\Past; \Future | \CausalState, \AlternateState]) \\
              &\quad = I[\Past; \Future] ~,
\end{align}
where we made use of the atoms that vanish. Thus, the four-way mutual
information simply reduces to the mutual information between the past
and the future---the excess entropy:
\begin{align}
I[\Past; \CausalState; \AlternateState; \Future] = \EE ~.
\end{align}
Similar calculations reduce the other information measures in
Fig.~\ref{fig:4variable} correspondingly. We now consider these reductions in
turn.

\subsection{Case: Minimal Unifilar Presentation}

The set of minimal unifilar presentations corresponds exactly to the $\eMs$,
up to state relabeling. The states in these presentations, the causal states,
induce a partition of the infinite pasts via the function $\epsilon(\past)$.

\begin{figure}
\centering
\ifcolor
    \includegraphics[width=\columnwidth]{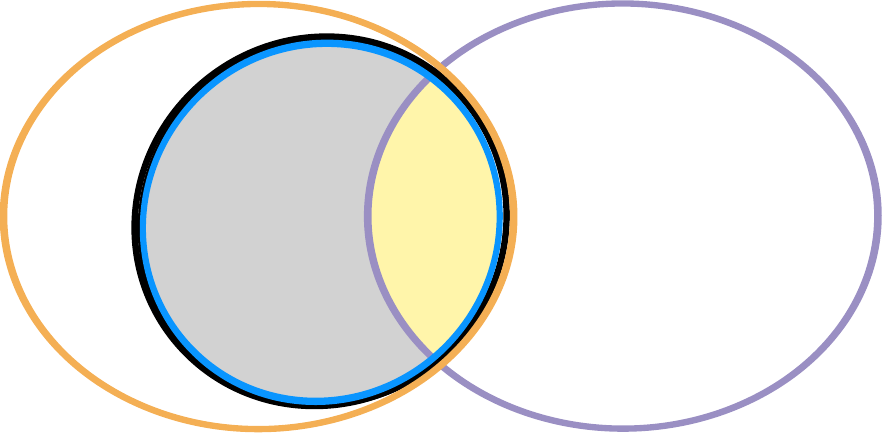}
    \includegraphics[scale=.95]{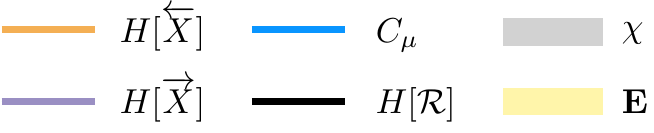}
\else
    \includegraphics[width=\columnwidth]{images/eM_gray}
    \includegraphics[scale=.95]{images/legend_small_gray}
\fi
\caption{The information diagram for an \eM. The states of the presentation
  are causal states and induce a partition on the past.
  The entropy over the states, $H[\AlternateState_0] = H[\CausalState_0]$,
  defines the statistical complexity ($\Cmu$). The process
  crypticity is the difference of the statistical complexity and the
  excess entropy.
  }
\label{fig:eM}
\end{figure}

The information diagram and entropy growth plot are particularly simple,
as seen in Fig.~\ref{fig:eM} and Fig.~\ref{fig:eMEntropyGrowth}. This simplicity
derives from the efficient predictive role the causal states play. Referring
to the I-diagram, $H[\CausalState | \Past] = 0$ because of determinism of the
$\epsilon(\past)$ map, Eq.~(\ref{Eq:PredictiveEquivalence}). Next, causal
states, as well as all other states we consider, are prescient states and so
$I[\Past ; \Future | \AlternateState] = 0$. These straightforward requirements
entirely determine the form of the \eM\ I-diagram in Fig.~\ref{fig:eM}. As we
step through the space of presentation classes, we will see these
relationships become more complex.

\begin{figure}
\centering
\ifcolor
    \includegraphics[width=\columnwidth]{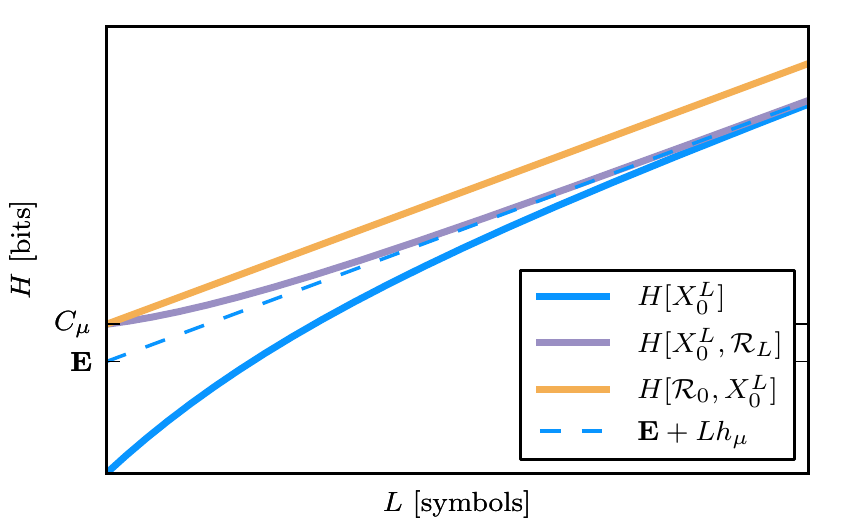}
\else
    \includegraphics[width=\columnwidth]{images/BEGem_gray}
\fi
\caption{Entropy growth for a generic \eM.
  \BE{L} and \BSE{L} both converge to the same asymptote.
  \SBE{L} is linear.
  }
\label{fig:eMEntropyGrowth}
\end{figure}

There are three quantities that require attention in this figure. First, the
state entropy $H[\AlternateState]$ is equal to $\Cmu$---the statistical
complexity. This particular state information is considered privileged as it
is the state information associated with the \eM\ and so the process.
The excess entropy $\EE$ is the mutual information
between the past and future and is also exactly that information which the
(causal) states contain about the future. Lastly, the crypticity $\PC$ is the
amount of information ``overhead'' required for prediction using the \eM.
Generally, this overhead is associated with the presentation as well as the
process \emph{itself}, due to the uniqueness of the \eM\ presentation. It
is the irreducible memory associated with the process. At any time, the process
itself or a predictive model must keep track of $\Cmu$ bits of state
information, while only $\EE$ bits of this information are correlated with
the future.

The entropy growth plot, Fig.~\ref{fig:eMEntropyGrowth}, is also simplified
by using causal states. In terms of our newly defined integrals:
$\mathcal{K}_n = 0$ for all $n$ and
$\mathcal{J}_1 = H[\CausalState] - \mathcal{I}_1 = \PC$.

A simple example that illustrates all of these points is provided by the Golden
Mean Process and its \eM; see Fig. \ref{fig:GMeM}. When the probability $p$ is
chosen to be $\frac{1}{2}$, the values of our information measures are
$\Cmu = \log_2 (3) - \frac{2}{3} = 0.9183$ bits, $\PC = \frac{2}{3}$ bits,
and $\EE = \Cmu - \PC = 0.2516$ bits. As we explore alternate presentations, we
will return to this process as a common thread for explanation and intuition.

\begin{figure}
\centering
\includegraphics[scale=.8]{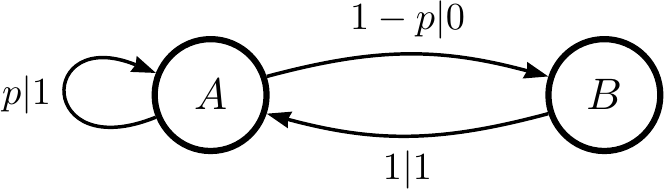}
\caption{The \eM\ presentation of the Golden Mean Process.}
\label{fig:GMeM}
\end{figure}

\subsection{Case: Weakly Asymptotically Synchronizable Presentations}

Let's relax the minimality constraint leaving the \eMs\ for presentations
that are nonminimal unifilar and weakly asymptotically synchronizable.
Again, the states correspond to a partition of the infinite pasts, but since
they are prescient and not minimal unifilar, the partition must be a refinement
of the causal-state partition~\cite{Shal98a}.

The effect of this is benign as seen in both the I-diagram (Fig. \ref{fig:AS})
and the entropy growth plot (Fig.~\ref{fig:ASEntropyGrowth}). In Fig.
\ref{fig:AS}, weakly asymptotically synchronizability ensures that
$H[\AlternateState | \Past] = 0$.
Demanding prescient states determines the form of the I-diagram.
Figure \ref{fig:AS} indicates that $H[\AlternateState] > H[\CausalState]$.
This is a consequence of
$\AlternateStateSet$ being a nontrivial refinement of $\CausalStateSet$.

\begin{figure}
\centering
\ifcolor
    \includegraphics[width= \columnwidth]{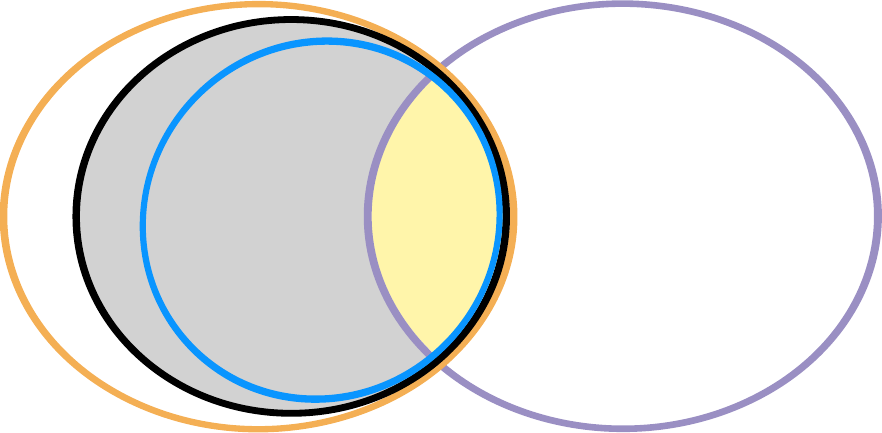}
    \includegraphics[scale=.95]{images/legend_small}
\else
    \includegraphics[width= \columnwidth]{images/AS_gray}
    \includegraphics[scale=.95]{images/legend_small_gray}
\fi
\caption{The information diagram for a presentation that is weakly
  asymptotically synchronizable, but not necessarily minimal unifilar.
  The states still induce a partition on the infinite past. The presentation
  crypticity $\PC(\AlternateState)$ is the difference of the state entropy
  $H[\AlternateState] \geq \Cmu$ and the excess entropy $\EE$.
  }
\label{fig:AS}
\end{figure}

Examining the entropy growth plot, the increased state information is reflected
in the values of the block-state and state-block entropy curves at $L=0$.
Additionally, it is interesting to note what happens to the cryptic order.
We generalized the definition of cryptic order to be that length where
the block-state entropy reaches its asymptote. Since block-state entropy is
nondecreasing, this suggests that it might be forced to reach its asymptote
at a larger value of $L$ than the cryptic order for the \eM\ presentation.
We can see that this is in fact true by expanding the following joint entropy
in two ways. Note that we combine variables from two \emph{different}
presentations and expand
$H[\MeasSymbol_0^L, \CausalState_L, \AlternateState_L]$:
\begin{align}
H[\MeasSymbol_0^L \AlternateState_L]
  & = H[\AlternateState_L | \MeasSymbol_0^L \CausalState_L]
  	+ H[\MeasSymbol_0^L \CausalState_L]
	- H[\CausalState_L | \MeasSymbol_0^L \AlternateState_L] \nonumber \\
  & = H[\AlternateState_L | \MeasSymbol_0^L \CausalState_L]
  	+ H[\MeasSymbol_0^L \CausalState_L] \nonumber \\
  & \geq H[\MeasSymbol_0^L \CausalState_L] ~.
\end{align}

\begin{figure}
\centering
\ifcolor
    \includegraphics[width=\columnwidth]{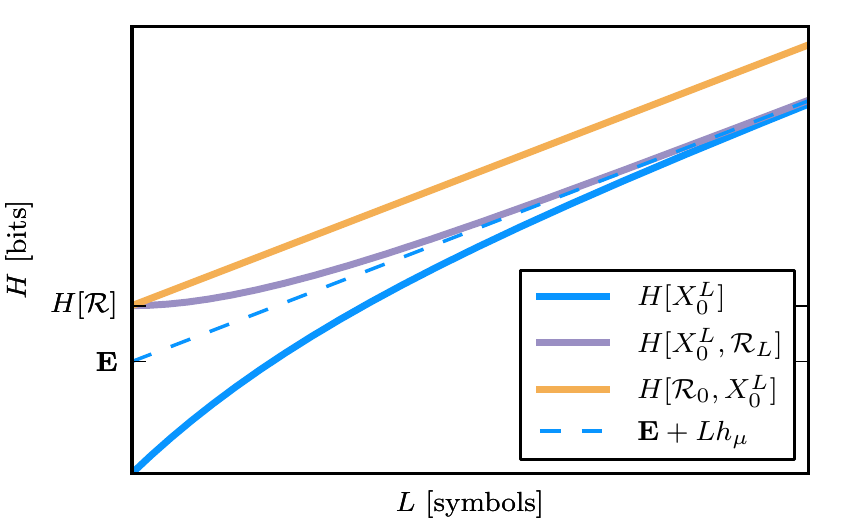}
\else
    \includegraphics[width=\columnwidth]{images/BEGas_gray}
\fi
\caption{Entropy growth for a weakly asymptotically synchronizing
presentation. \BE{L} and \BSE{L} both converge to the same
asymptote.  \SBE{L} is linear.
$H[\AlternateState]$ is larger than $\Cmu$.}
\label{fig:ASEntropyGrowth}
\end{figure}

In the above, we make use of the fact that $\AlternateStateSet$ is a refinement
of $\CausalStateSet$ and that conditional entropies are positive semi-definite.
This shows that the block-state curve for the nonminimal presentation lies
above or on the curve for the \eM\ presentation. Since block and block-state
entropies share an asymptote---$\EE + L \hmu$---the nonminimal unifilar
block-state entropy will reach
its asymptote at a value greater than or equal to the process's cryptic order.
More care will be required in the subsequent cases, as the
relations among entropy growth functions are more complicated.

To illustrate these class characteristics, consider the following three-state
presentation of the Golden Mean Process in Fig.~\ref{fig:GMAS}. The original
causal state partition, $\{A = *1, B = *0\}$, has become refined. (Here, $*$
denotes any allowed history.) We now have $\{A = *11, B = *0, C = *01\}$. It
is straightforward to verify that $H[\AlternateState] = \log_2 (3) = 1.585$ bits.
Excess entropy is unchanged as it is a feature of the process language and not
the presentation. As illustrated in Fig.~\ref{fig:AS}, the crypticity grows
commensurately with $H[\AlternateState]$.

We have shown that for weakly asymptotically synchronizable presentations the
presentation cryptic order generally will be larger than the cryptic order.
It is interesting to note that it is also possible for the presentation cryptic
order to surpass even the Markov order. Our three-state example
(Fig.~\ref{fig:GMAS}) is \cryptic{2} while the Markov order remains $R = 1$
as it also depends only on the process language.

Since the Markov order $R$ bounds the domain of the $\mathcal{I}_0$ integral
and the presentation cryptic order $k$ bounds the domain of the
$\mathcal{J}_0$ integral, the domain of the synchronization information
is bounded by $\max\{R, k\}$.

\begin{figure}
\centering
\includegraphics[scale=.8]{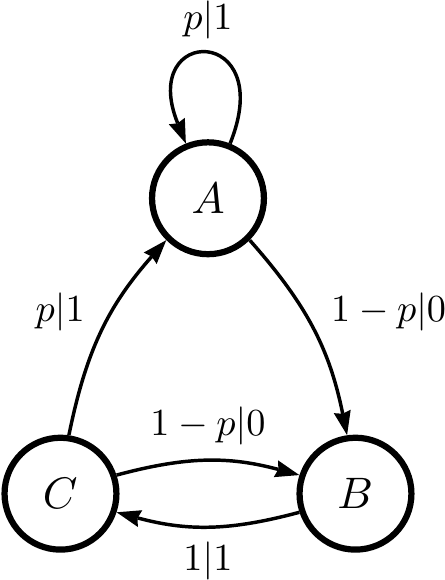}
\caption{A weakly asymptotically synchronizable and nonminimal unifilar
  presentation of the Golden Mean Process: observing a $0$ synchronizes
  the observer to state $B$.
  }
\label{fig:GMAS}
\end{figure}

\subsection{Case: Unifilar Presentations}

Removing the requirement that a presentation be weakly asymptotically
synchronizable, we no longer operate with (recurrent) states that correspond
to a partition of the infinite past, but rather to a covering of the set of
infinite pasts. That is, $\eta(\past)$ can be multivalued, although for
each $\alternatestate \in \AlternateStateSet$, $\eta^{-1} (\alternatestate)$
is a set of pasts that is a subset of some causal state's set of pasts.

Every allowable infinite history induces at least one state in the
presentation---this is the definition of an allowable infinite history.
Additionally, any presentation that is not weakly asymptotically
synchronizable must have a (positive measure) set of histories where
each history induces more than one state.

Consider a unifilar presentation and an infinite history which induces only one
state. Due to unifilarity, we can use this history to construct an infinite set
of histories that are also synchronizing. We conjecture that this set of
histories must have zero measure and, even stronger, that for finite-state
unifilar presentations with a single recurrent component, there are no
synchronizing histories.

This inability to synchronize, a product of the nontrivial covering, is
represented as the information measure $\GI$ in Fig.~\ref{fig:Unif}.
This information is not captured by the causal states.
In fact, it is not even captured by the past (or the future).
It also is not necessary for making predictions with the same power as the \eM.
Like $\PC(\AlternateState)$, $\GI$ is unnecessary for prediction.
However, unlike $\PC(\AlternateState)$, $\GI$ does not capture any structural property
of the process. Instead, it represents degrees of freedom entirely decoupled
(informationally) from the process language and prediction. For this reason,
we called it the \emph{gauge information}.

\begin{figure}
\centering
\ifcolor
    \includegraphics[width=\columnwidth]{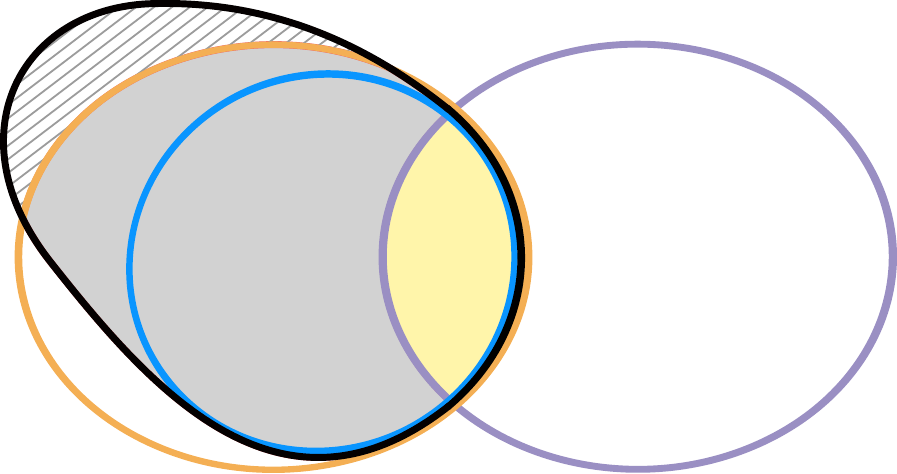}
    \includegraphics[scale=.95]{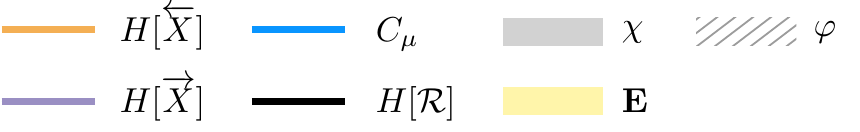}
\else
    \includegraphics[width=\columnwidth]{images/Unif_gray}
    \includegraphics[scale=.95]{images/legend_medium_gray}
\fi
\caption{The information diagram for a presentation that is not weakly
  asymptotically synchronizable, but still unifilar. The states are
  prescient, but no longer induce a partition on the infinite past.
  Furthermore, the states contain information that the past does not contain.
  The presentation crypticity is the difference of the state entropy
  $H[\AlternateState_0] \geq \Cmu$ and the excess entropy
  $\EE = I[\protect\Past_0;\protect\Future_0]$.
  }
\label{fig:Unif}
\end{figure}

The entropy growth plot of Fig.~\ref{fig:UnifEntropyGrowth} has a new and
significant feature representing the change in class. The asymptotes of the
block entropy and block-state entropy become nondegenerate. This has the effect
of making the synchronization information diverge. Although this fact follows
immediately from the definition of weakly asymptotically synchronizable,
it is instructive to see its geometric representation.

\begin{figure}
\centering
\ifcolor
    \includegraphics[width=\columnwidth]{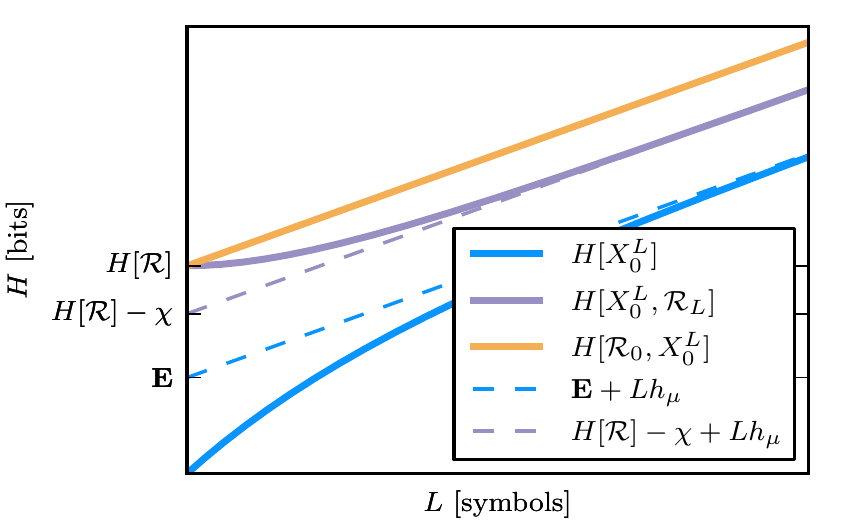}
\else
    \includegraphics[width=\columnwidth]{images/BEGu_gray}
\fi
\caption{Entropy growth for a not weakly asymptotically synchronizable, but
  unifilar presentation.
  \BE{L} and \BSE{L} both converge to different asymptotes.
  \SBE{L} is linear.  $H[\AlternateState]$ is larger than $\Cmu$.
  }
\label{fig:UnifEntropyGrowth}
\end{figure}

Since, from this point forward, synchronization information is always infinite,
we find it necessary to re-express what synchronization information means.
It can be denoted, recall Eq.~(\ref{eq:SyncInfoWhenInfinite}), as the sum of a
finite piece and the limit of a linear (in $L$) piece:
$\SI = \mathcal{I}_0 + \mathcal{J}_0 + \lim_{L \to \infty}{L \GI}$.
This rate of increase of the linear piece is exactly the gauge information.

It is also interesting to note that when this information is obtained---that
is, a constraint is imposed upon the descriptional degrees of
freedom---unifilarity maintains synchronization as more data is produced. In
this sense, acquiring gauge information is a ``one-time'' cost.

The Golden Mean Process presentation in Fig.~\ref{fig:GMU} illustrates all of
the features described above. It is straightforward to see that this
presentation is not weakly asymptotically synchronizing. Any history, finite or
infinite, has exactly two states that it induces. This degeneracy is never
broken, due to unifilarity. Rephrasing, the gauge information value, $\GI = 1$
bit, derives from the fact that each infinite history induces one of two states
with equal likelihood. This relies on the fact that there is no oracular
information contribution---$\OI = 0$ bits since the presentation is
unifilar---to disentangle from the gauge information.

\begin{figure}
\centering
\includegraphics[scale=.8]{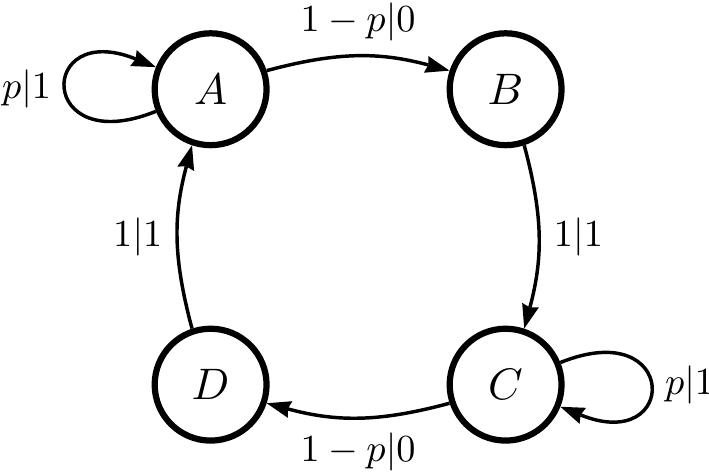}
\caption{A unifilar, but not weakly asymptotically synchronizing,
  presentation of the Golden Mean Process.
  }
\label{fig:GMU}
\end{figure}

\subsection{Case: Nonunifilar Presentations}
\label{sec:Nonunifilar}

Finally, we remove the requirement of unifilarity and examine the much larger,
complementary space of nonunifilar presentations. Only one nonunifilar state
must be present to change the class of the whole presentation. This ease of
breaking unifilarity is why nonunifilar presentations form a much larger class.

Examining the I-diagram in Fig.~\ref{fig:NUnif}, we notice one new feature:
the oracular information $\OI = I[\AlternateState ; \Future | \Past]
= I[\AlternateState ; \Future | \CausalState] \ne 0$. The oracular information
is a curious quantity and so deserves careful interpretation. It is the degree
to which the presentation state reduces uncertainty in the future beyond that
for which the past can account. One might think of this feature as
``super-prescience''. Not only is the information from the past being maximally
utilized for prediction, but some additional information is also injected. We
make several remarks about this.

\begin{figure}
\centering
\ifcolor
    \includegraphics[width=\columnwidth]{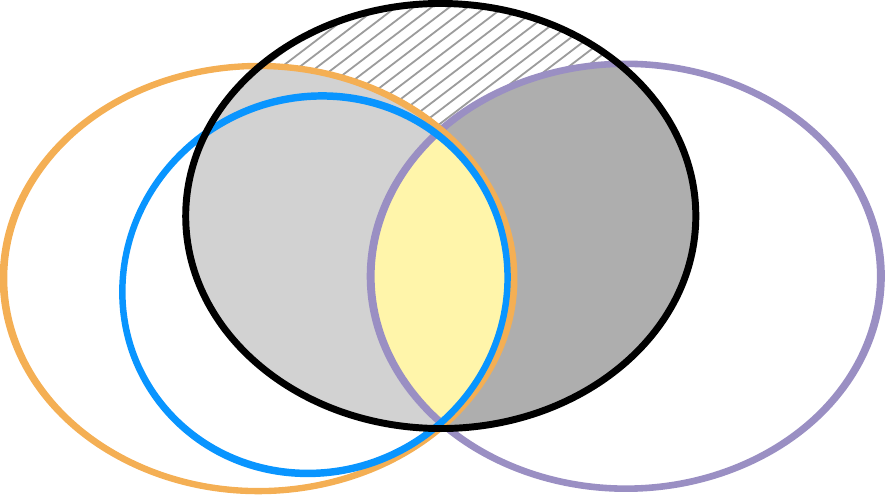}
    \includegraphics[scale=.95]{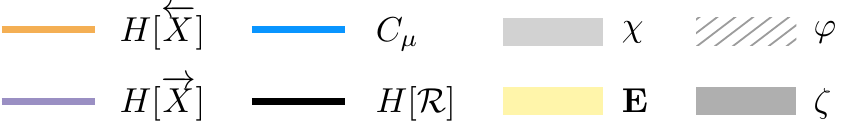}
\else
    \includegraphics[width=\columnwidth]{images/Nunif2_gray}
    \includegraphics[scale=.95]{images/legend_gray}
\fi
\caption{The information diagram for a presentation that is not unifilar.
  The states are super-prescient, do not induce a partition on the past, and
  have information not contained in the past. The presentation crypticity
  is the difference of the state entropy $H[\AlternateState]$ and the excess
  entropy ($\EE$). Note, the state entropy can also be smaller than the
  statistical complexity.
  }
\label{fig:NUnif}
\end{figure}

It is well known that a process's nonunifilar presentations may be smaller than
the corresponding \eM. This fact is sometimes cited~\cite{Crut92c} as providing
evidence that the smaller nonunifilar presentation is the more ``natural''
one \footnote{Similar observations appeared recently; for example, see Ref.
\cite{Loeh09a}. In a sequel we compare this to the earlier results of Refs.
\cite{Crut92c,Uppe97a}.}.
While it is true that the state information $H[\AlternateState]$ can be smaller
than $\Cmu$, and in fact often is, the I-diagram makes plain the fact that
oracular information must be introduced to determine $\AlternateState$ and,
thus, make a super-prescient prediction. For this reason, unless one is
transparent about allowing for oracular information, it is not appropriate
to make a judgment about naturalness of nonunifilar presentations.

Given that we do not have the luxury of access to an oracle, we might like to
know how these presentations perform without this information. The nonoracular
part of $I[\AlternateState ; \Future]$ is simply $\EE$. That is, without the
oracular information, we predict just as we would with any other prescient
presentation. However, the predictions are made using \emph{distributions over
states} rather than individual states. (The former are the mixed states of
Ref. \cite{Crut08b}.) More importantly, as we continue to make predictions,
the state distribution evolves through a series of distributions. These
distributions are in 1-to-1 correspondence with the causal states of the \eM.
And so, for a nonoracular user of a nonunifilar presentation to communicate her
history-induced state to another requires the transmission of $\Cmu$ bits.
The statistical complexity is inescapable as the (nonoracular) information
storage of the process.

When discussing nonweakly asymptotically synchronizable, but unifilar
presentations, we indicated that the gauge information was a ``one-time'' cost.
Now, we ask the same question of the two informations---gauge and
oracular---that are not products of the past. Since we no longer have
unifilarity, state uncertainty is dynamically reintroduced as synchronization
is lost. That is, nonunifilar presentations are allowed to \emph{locally}
resynchronize following the introduction of state uncertainty. The net result
is that over time synchronization is repeatedly lost and reacquired.

\begin{figure*}
\centering
\ifcolor
    \includegraphics[width=\textwidth]{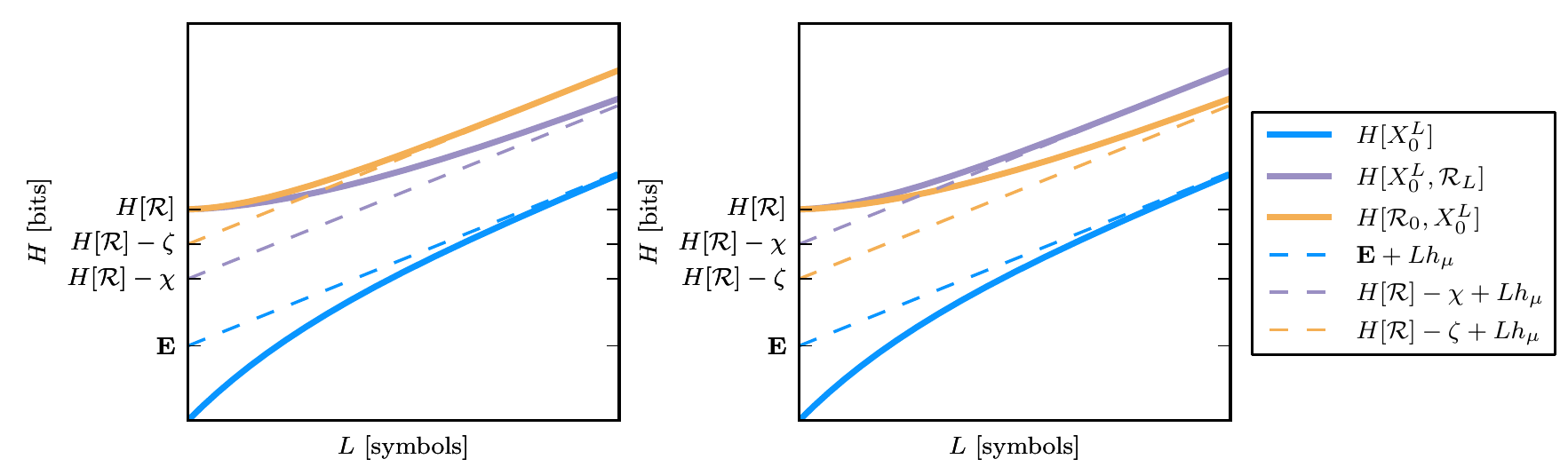}
\else
    \includegraphics[width=\textwidth]{images/BEGnu_nu2_gray}
\fi
\caption{Entropy growth for a nonunifilar presentation.  Left: \BE{L} and
  \BSE{L} both converge to different asymptotes; \SBE{L} is not linear and
  $H[\AlternateState]$ is larger than $\Cmu$. Right: The same as on the left,
  but illustrating that $\PC$ can be less than $\OI$.
  }
\label{fig:NUnifEntropyGrowth}
\end{figure*}

The entropy growth plot of Fig.~\ref{fig:NUnifEntropyGrowth} makes one last
adjustment to acknowledge the change in class. For the first time, the
state-block entropy is nonlinear. It approaches its asymptote from above and,
moreover, the asymptote is independent of the block-state asymptote. The
projection back onto the y-axis mirrors our final and most general I-diagram
of Fig.~\ref{fig:NUnif}. The left panel emphasizes that the crypticity
$\PC(\AlternateState)$ can be less than the oracular information $\OI$, in
general cases.

A nonunifilar presentation of the Golden Mean Process is shown in
Fig.~\ref{fig:GMNU}. All of the above-mentioned quantities are nonzero for
this presentation: For $p = 1/2$, the crypticity $\PC(\AlternateState) = 1/3$
bits, the gauge information $\GI = 1$ bit, and the oracular information
$\OI = 1/3$ bits. The value of the gauge information ($1$ bit) is easy to
understand. It indicates that the nonunifilar presentation is two copies of a
unifilar presentation of the Golden Mean Process sutured together. All of
history space is covered twice and the choice of which component of the cover
is visited is a fair coin flip. The crypticity and oracular information
(crypticity's time-reversed analog) are the same, due to the nonunifilar
presentation respecting the time-reverse symmetry of the Golden Mean Process
\cite{Crut08b}.

\begin{figure}
\centering
\includegraphics[scale=.8]{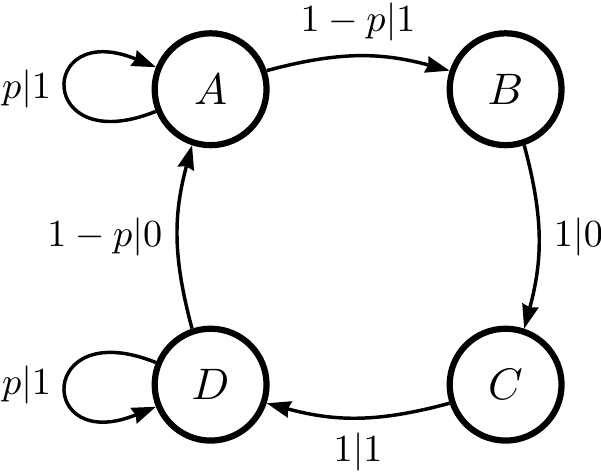}
\caption{A nonunifilar presentation of the Golden Mean Process.}
\label{fig:GMNU}
\end{figure}

\section{Conclusions}

Our development started out discussing synchronization and control. The tools
required to address these---the block-state and state-block entropies---quickly
led to a substantially enlarged view of the space of competing models, the rival
presentations, and a new collection of information measures that reflect their
subtleties and differences.

As milestones along the way, we gave example presentations of the well known
Golden Mean Process that went from the \eM\ to a nonminimal nonsynchronizing
nonunifilar presentation. Table ~\ref{tab:GMP_Presentation_IMeasures} summarizes
the quantitative results. It gives the entropy rate $\hmu$, statistical
complexity $\Cmu$, excess entropy $\EE$, and the crypticity $\PC$ for the
process itself. Immediately following, it compares the analogous measures for
the range of presentations considered. In addition, the gauge information $\GI$
and the oracular information $\OI$, being properties of presentations, are
added. Careful study
of the table shows how the measures track the presentations' structural changes.

A few comments are in order about the tools the development required. The first
were the block-state and state-block entropies, as noted. Analyzing their
word-length convergence properties was the backbone of the approach---one
directly paralleling the previously introduced entropy convergence
hierarchy~\cite{Crut01a}. Another important tool was the I-diagram. While it
is not necessary in establishing final results, it is immensely helpful in
organizing one's thinking and in managing the complications of multivariate
information measures. Methodologically speaking, the principal subject was the
four-variable---past, future, causal state, and presentation state---I-diagram
with its sigma algebra of 15 atoms. Thus, the methodology of the development
turned on just two tools---block entropy convergence and presentation
information measures.

\begin{table*}[tbp]
\setstretch{1.5}
\centering
\textbf{Information Measures for Alternative Presentations}
\\[.8em]
\begin{tabular}{|l||c|c|c|c|c|c|}
\cline{1-5}
Process & $\hmu$ & $\Cmu$ & $\EE$ & $\PC$ &
\multicolumn{2}{c}{ \multirow{2}{*}{} }
\strut\\
\cline{1-5}
~~~Golden Mean & $2/3$ & $\log_2(3) - 2/3$ & $\log_2(3) - 4/3$ & $2/3$ &
\multicolumn{2}{c}{}
\strut\\
\cline{1-5}\noalign{\vskip1.5mm}
\hline
Presentation &
$H[\MeasSymbol|\AlternateState]$ &
$H[\AlternateState]$ &
$I[\AlternateState;\Future]$ &
$\PC(\AlternateState)$ &
~~$\GI$~~ &
~~$\OI$~~ \strut\\
\hline
~~~\EM & $\hmu$ & $\Cmu$ & $\EE$ & $\PC$ & 0 & 0 \\
~~~Synchronizable & $\hmu$ & $\log_2(3)$ & $\EE$ & $4/3$ & 0 & 0 \\
~~~Unifilar & $\hmu$ & $\log_2(3) + 1/3$ & $\EE$ & 5/3 & 1 & 0 \\
~~~Nonunifilar & 1/3 & $\log_2(3) + 1/3$ & $\log_2(3) -1$ & 1/3 & 1 & 1/3 \\
\hline
\end{tabular}
\setstretch{1.1}
\caption{Comparison of information measures for presentations of the
  Golden Mean Process with transition parameter $p = 1/2$.
  }
\label{tab:GMP_Presentation_IMeasures}
\end{table*}

As for the concrete results, we showed that there are two mechanisms operating
in processes that are hard to synchronize to, as measured by the synchronization
information which consists of two corresponding independent contributions. The
first is the
transient information which reflects entropy-rate overestimates that occur at
small block lengths. The second, new here, reflects the state information that
is not retrodictable using the future. With these two contributions laid out,
the general connection between synchronization and transient information,
previously introduced in Ref.~\cite{Crut01a}, became clear. We also pointed out
that the synchronization information for nonsynchronizing presentations can
diverge. This, in turn, called for a generalized definition of synchronization
appropriate to all presentations.

We also generalized the process crypticity, beyond the domain of \eM\ optimal
presentations, to describe the amount of presentation state information that
is shared with the past but not transmitted to the future. A sibling of the
crypticity, we introduced a new information measure for generic
presentations---the oracular information---that is the amount of state
information shared with the future, but not derivable from the past.

Finally, to account for ``components'', either explicitly or implicitly included
in a presentation, that are not justified by the process statistics, we
introduced the gauge information, intentionally drawing a parallel to the
concept of gauge degrees-of-freedom familiar from physics.

One immediate result was that the information measures allowed us to delineate
the hierarchy of a process's presentations. The hierarchy goes from the unique, minimal unifilar, optimal predictor (\eM) to nonminimal unifilar, weakly
asymptotically synchronizing presentations to nonsynchronizing, unifilar
presentations. We showed these are nested classes. Stepping outside to the
nonunifilar presentations leaves one in a markedly larger class for which
all of the information measures play a necessary role.

We trust that the presentation hierarchy makes the singular role of the
\eM\ transparent. First, the \eM's minimality and uniqueness are those of the
corresponding process. This cannot be said for alternative presentations.
Second, there is a wide range of properties that can be efficiently calculated,
when alternative presentations may preclude this. One cannot calculate a
process's stored information ($\Cmu$) or information production rate ($\hmu$)
from, for example, nonunifilar presentations. The latter must be converted,
either directly or indirectly, to the process's \eM\ to calculate them.

Nonetheless, as discussed at some length in Ref.~\cite{Crut92c}, in varying
circumstances---limited material, inference, or compute-time resources; ready
access to sources of ideal randomness; noisy implementation substrates; and
the like---the \eM\ may not be how an observer should model a process.
A minimal nonunifilar presentation, that is necessarily more stochastic
internally than the \eM~\cite{Shal98a}, may be preferred due to it having a
smaller set of states.

Recalling the duality of synchronization and control, we close by noting that
essentially all of the results here apply to the setting in which an agent
attempts to steer a process into desired states. The efficiency with which the
control signals achieve this is reflected in the analogue of block entropy
convergence. The very possibility of control has its counterparts in an
implementation hierarchy that mirrors the presentation hierarchy, but with
controllability instead of synchronizability.

\appendix

\section{Notation Change for Total Predictability}
\label{RURODiff}

The definition for $\mathcal{I}_n$ in Eq.~(\ref{eq:IntegralOperator})---the
total predictability---represents a minor change in notation from
Ref.~\cite{Crut01a}. (We
refer to the latter as RURO, abbreviating its title.) There, the minimum
$L$ was usually $n$ except for $n=2$, when the minimum $L$ value was $L=1$
instead. One reason for the change in definition is that $\mathcal{I}_2$ now
does not depend on any assumption (prior) for symbol entropy rate and depends
\emph{only} on asymptotic properties of the process.

To make this explicit, note that the original definition of total
predictability contained a boundary term:
\begin{align}
 \mathbf{G}_\text{RURO} &= \Delta^2 H(1) + \sum_{L=2} \Delta^2 H(L) ~,
\intertext{where}
 \Delta^2 H(1) &= h_\mu(1) - h_\mu(0) = H(1) - \log_2|\MeasAlphabet| ~.
\end{align}
The logarithm term characterized the entropy rate estimate before any
probabilities are considered. In the modified definition of total
predictability, we drop the boundary term, giving:
\begin{align}
 \mathbf{G} &\equiv \mathcal{I}_2 = \sum_{L=2} \Delta^2 H(L) ~.
\end{align}
The two quantities are related by:
\begin{align}
  \mathbf{G}_\text{RURO}
  &= \mathbf{G} + \Delta^2 H(1) \\
  &= \mathbf{G} + H(1) - \log_2 |\MeasAlphabet| ~.
\end{align}
This affects relationships involving $\mathbf{G}$. Previously, for example,
\begin{align}
  \mathbf{G}_\text{RURO} &= - \mathbf{R} \leq 0 ~,
\end{align}
where $\mathbf{R}$ is the total redundancy. Now,
\begin{align}
  \mathbf{G} &= -\mathbf{R} - \Delta^2 H(1)\\
             &= \log_2|\MeasAlphabet| - H(1) - \mathbf{R}  ~.
\end{align}

\section{State-Block Entropy Rate Estimate}
\label{sbelowerboundhmu}

In this section, we prove Thm.~\ref{sbeconvergence}, which states
that $H[\MeasSymbol_L | \AlternateState_0^{\vphantom{L}}, \MeasSymbol_0^L]$
converges monotonically (nondecreasing) to the entropy rate.

\begin{proof}
First, we show that difference in the $\AltSBE{L}$ forms a nondecreasing
sequence:
\begin{align}
&   H[\MeasSymbol_{L-1} | \AlternateState_0^{\vphantom{L}},
                           \MeasSymbol_0^{L-1}] \\
&\quad=
    H[\MeasSymbol_{L} | \AlternateState_1^{\vphantom{L}},
                        \MeasSymbol_1^{L-1}] \\
&\quad=
    H[\MeasSymbol_{L} | \AlternateState_0^{\vphantom{L}},
                        \MeasSymbol_0^{\vphantom{L}},
                        \AlternateState_1^{\vphantom{L}},
                        \MeasSymbol_1^{L-1}] \\
&\quad\leq
    H[\MeasSymbol_{L} | \AlternateState_0^{\vphantom{L}},
                        \MeasSymbol_0^{\vphantom{L}},
                        \MeasSymbol_1^{L-1}] \\
&\quad=
    H[\MeasSymbol_{L} | \AlternateState_0^{\vphantom{L}},
                        \MeasSymbol_0^{L}] ~.
\end{align}
Next, we show this sequence is bounded and, thus, has a limit.
For all $k \geq 0$, we have:
\begin{align}
&    H[\MeasSymbol_L | \AlternateState_0^{\vphantom{L}}, \MeasSymbol_0^L] \\
&\quad = H[\MeasSymbol_L | \MeasSymbol_{-k}^{k},
                           \AlternateState_0^{\vphantom{L}}, \MeasSymbol_0^L]\\
&\quad \leq H[\MeasSymbol_L | \MeasSymbol_{-k}^{L+k}] \\
&\quad = H[\MeasSymbol_{L+k} | \MeasSymbol_0^{L+k}]~.
\end{align}
Since this holds for all $k$, it also holds in the limit as $k$ tends to
infinity, which is the definition of the entropy rate.  Thus,
$H[\MeasSymbol_L | \AlternateState_0^{\vphantom{L}}, \MeasSymbol_0^L]$
is a nondecreasing sequence and bounded above by $\hmu$.

Finally, we show that this bounded sequence converges to $\hmu$. To do this,
we will show that the difference
\[
H[\MeasSymbol_L^{\vphantom{L}} | \MeasSymbol_0^L] -
H[\MeasSymbol_L^{\vphantom{L}} |
    \AlternateState_0^{\vphantom{L}}, \MeasSymbol_0^L]
 =
  I[\AlternateState_0^{\vphantom{L}}; \MeasSymbol_L^{\vphantom{L}}
    | \MeasSymbol_0^L]
\]
converges to zero.  Then, since the first term (differences in the
block entropies) is known to converge to the entropy rate, the claim will
be proved. We have:
\begin{align}
 H[\AlternateState_0]
& \geq \lim_{L \rightarrow \infty}
  I[\AlternateState_0^{\vphantom{L}}; \MeasSymbol_0^{L},
    \MeasSymbol_L^{\vphantom{L}}] \\
&=
  \lim_{L \rightarrow \infty} \sum_{k=0}^L
  I[\AlternateState_0^{\vphantom{L}}; \MeasSymbol_k^{\vphantom{L}}
    | \MeasSymbol_0^k] ~.
\end{align}
Since the sum is finite, the terms must tend to zero.
\end{proof}

\section{Reducing the Presentation I-Diagram}
\label{zeroatoms}

Proving that the various multivariate information measures vanish
makes use of a few facts about states:
\begin{itemize}
	\item $H[\CausalState | \Past] = 0$.
	\item $H[\Past; \Future | \CausalState] = 0$.
	\item $I[\Past; \Future | \AlternateState] = H[\Future |
	\AlternateState] - H[\Future | \AlternateState, \Past] = 0$.
\end{itemize}
The last one follows from limiting ourselves to states that actually
generate the process. Thus, additional conditioning on the past cannot
influence the future, as the current state alone determines the future.

The following atoms vanish:
\begin{itemize}
	\item $H[\CausalState | \Past , \AlternateState , \Future]$:
	\begin{align*}
		H[\CausalState | \Past , \AlternateState , \Future]
		\leq H[\CausalState | \Past] = 0 ~.
	\end{align*}

	\item $I[\CausalState ; \AlternateState | \Past , \Future]$:
	\begin{align*}
		I[\CausalState ; \AlternateState | \Past , \Future] & =
		H[\CausalState | \Past, \Future] - H[\CausalState | \Past
		, \AlternateState , \Future] \\
		& = H[\CausalState | \Past, \Future] - 0 \\
		& \leq H[\CausalState | \Past] \\
		& = 0 ~.
	\end{align*}

	\item $I[\CausalState ; \AlternateState ; \Future | \Past]$:
	\begin{align*}
		I[\CausalState ; \AlternateState ; \Future | \Past] & =
		I[\CausalState; \AlternateState | \Past] - I[\CausalState
		; \AlternateState | \Past , \Future] \\
		& = I[\CausalState; \AlternateState | \Past] - 0 \\
		& = H[\CausalState | \Past] - H[\CausalState |
		\AlternateState, \Past] \\
		& = 0 - H[\CausalState | \AlternateState, \Past] ~.
	\end{align*}
	Finally, note that
	\begin{align*}
		|H[\CausalState | \AlternateState, \Past]| & \leq |H[\CausalState | \Past]| \\
		& = 0 ~.
	\end{align*}

	\item $I[\CausalState ; \Future | \Past , \AlternateState]$:
	\begin{align*}
		I[\CausalState ; \Future | \Past , \AlternateState]
		& = H[\CausalState | \Past , \AlternateState] -
		H[\CausalState | \Past, \Future, \AlternateState] \\
		& = H[\CausalState | \Past , \AlternateState] - 0 \\
		& \leq H[\CausalState | \Past] \\
		& = 0 ~.
	\end{align*}

	\item $I[\Past ; \Future | \CausalState , \AlternateState]$:
	\begin{align*}
		I[\Past ; \Future | \CausalState , \AlternateState]
		& = H[\Future | \CausalState, \AlternateState] - H[\Future
		| \CausalState, \AlternateState, \Past] \\
		& = H[\Future | \CausalState, \AlternateState] - H[\Future
		| \CausalState, \AlternateState] \\
		& = 0 ~.
	\end{align*}

	\item $I[\Past ; \AlternateState ; \Future | \CausalState]$:
	\begin{align*}
		I[\Past ; \AlternateState ; \Future | \CausalState]
		& = I[\Past; \Future | \CausalState] - I[\Past; \Future |
		\CausalState, \AlternateState] \\
		& = 0 ~.
	\end{align*}

	\item $I[\Past ; \CausalState ; \Future | \AlternateState]$:
	\begin{align*}
		I[\Past ; \CausalState ; \Future | \AlternateState]
		& = I[\Past; \Future | \AlternateState] - I[\Past; \Future
		| \CausalState, \AlternateState] \\
		& = 0 ~.
	\end{align*}
\end{itemize}

The first four vanish due to the causal states being a function of the past.
The last three vanish since any presentation that generates the process
captures all the information shared between past and future.

\section*{Acknowledgments}

This work was partially supported by the DARPA Physical Intelligence Program.
The authors thank Dave Feldman, Nick Travers, and Luke Grecki for helpful
comments on the manuscript.

\bibliography{ref,chaos}

\end{document}